\documentclass[
 reprint,
 superscriptaddress,
 amsmath,amssymb,
 aps,
 floatfix
]{revtex4-2}

\usepackage{graphicx, amsfonts, amsmath, amssymb, physics, bm, amsthm}
\usepackage{xcolor}
\usepackage[hidelinks]{hyperref}
\usepackage{cleveref}
\usepackage{quantikz}
\usepackage{stmaryrd}
\usepackage{amsmath}
\usepackage{subfigure}

\usepackage{thm-restate}

\newtheorem{theorem}{Theorem}
\newtheorem{definition}{Definition}
\newtheorem{conjecture}{Conjecture}
\newtheorem{lemma}{Lemma}
\newtheorem{corollary}{Corollary}

\newcommand{\poly}{\operatorname{poly}}

\newcommand{\aqa}{$\langle aQa ^L\rangle $ Applied Quantum Algorithms, Universiteit Leiden, Leiden, Netherlands}

\newcommand{\liacs}{LIACS, Universiteit Leiden, Leiden, Netherlands}
\newcommand{\cern}{European Organization for Nuclear Research (CERN), Geneva 1211, Switzerland}

\begin{document}

\title{Quantum Advantage in Learning Quantum Dynamics via Fourier coefficient extraction}

\author{Alice Barthe}
\affiliation{\cern}
\affiliation{\aqa}
\affiliation{\liacs}

\author{Mahtab Yaghubi Rad}
\affiliation{\aqa}
\affiliation{\liacs}

\author{Michele Grossi}
\affiliation{\cern}

\author{Vedran Dunjko}
\affiliation{\aqa}
\affiliation{\liacs}

\date{\today}

\newcommand{\ab}[1]{{\color{blue} AB: #1}}
\newcommand{\id}{\mathbbm{1}}
\newcommand{\cca}{PQC-based functions}
\newcommand{\ccb}{Hamiltonian dynamics}
\newcommand{\subroutine}{Fourier coefficient extraction}

\begin{abstract}
One of the key challenges in quantum machine learning is finding relevant machine learning tasks with a provable quantum advantage. 
A natural candidate for this is learning unknown Hamiltonian dynamics. 
Here, we tackle the supervised learning version of this problem, where we are given random examples of the inputs to the dynamics as classical data, paired with the expectation values of some observable after the time evolution, as corresponding output labels. The task is to replicate the corresponding input-output function.
We prove that this task can yield provable exponential classical-quantum learning advantages under common complexity assumptions in natural settings.
To design our quantum learning algorithms, we introduce a new method, which we term \textit{\subroutine}~algorithm for parametrized circuit functions, and which may be of independent interest.
Furthermore, we discuss the limitations of generalizing our method to arbitrary quantum dynamics while maintaining provable guarantees.  
We explain that significant generalizations are impossible under certain complexity-theoretic assumptions, but nonetheless, we provide a heuristic kernel method, where we trade-off provable correctness for broader applicability.
\end{abstract}

\maketitle

\section{Introduction}
Identifying when quantum computers provide an advantage in learning tasks is a central challenge of quantum machine learning. 
Problems involving quantum systems are natural candidates, but many of their associated learning tasks can still be solved efficiently by classical computers. 
For instance, \cite{huang_provably_2022} shows that predicting ground state properties within a phase is classically tractable with limited training data.
Nonetheless, the prevailing intuition is that access to a device that can efficiently simulate the target quantum system should be advantageous. 
While earlier works identified contrived problems exhibiting such separation \cite{gyurik_exponential_2023,molteni_exponential_2024}, identifying quantum-classical learning separations in natural settings remains a compelling open question.

We address this challenge by considering the supervised learning problem of \emph{learning unknown Hamiltonian dynamics} from classical data. 
Specifically, we define a family of functions $f_\alpha(x) = \langle\psi(x,\alpha)|O|\psi(x,\alpha)\rangle$, where $|\psi(x,\alpha)\rangle = e^{iH(x,\alpha)t}|\psi_0\rangle$ is a time-evolved quantum state under a parametrized Hamiltonian $H(x,\alpha)$ for some fixed time $t$.
The task is to learn $f_\alpha$ (in the PAC sense, which we explain shortly) from a dataset $\mathcal{D} = \{(x_j, f_{\alpha}(x_j))\}_j$ for some fixed, unknown $\alpha$.
Our settings assume a more restricted (and maybe more realistic) access to the dynamics than found in related works in the literature \cite{caro_out--distribution_2023,wiebe_hamiltonian_2014}, in particular excluding the use of Hamiltonian learning methods.

Our contributions can be summarized as follows. 
We introduce a quantum subroutine to extract the Fourier decomposition of PQC-based functions, which we call \textit{\subroutine}~subroutine.
We then use this subroutine to define a quantum learning algorithm to solve the unknown Hamiltonian dynamics problem, which we formalize as the \ccb~concept class (defined shortly) in the probably approximately correct (PAC) framework.
In our approach, we use the fact that the quantum dynamics class can be approximated via Hamiltonian simulation by a class of functions built around parametrized quantum circuits, which we call the \cca~concept class. The latter class will serve as the effective hypothesis family of our learning algorithm.
Our learning algorithm is then proven to be correct for both concept classes. 
The algorithm is efficient when the number of unknown parameters (of the PQC or Hamiltonian) scales logarithmically with the system size. 
We also prove that no classical algorithm can solve this learning task under complexity-theoretic assumptions, yielding a separation for this natural problem.
We also consider the much more general setting of polynomially many unknown parameters. 
For this case, we analyze the potential and limitations of generalization of our method for this broader class of quantum dynamics problems. We identify conditional no-gos for provably efficient learners, but also propose a heuristic method for the problem, which may work in cases beyond what can be proven analytically.

This paper is structured as follows. \Cref{sec:background} covers the required background in PAC learning and PQCs. \Cref{sec:freqsamp} introduces the \textit{\subroutine} algorithm. \Cref{sec:learningpqc} applies this to learn PQCs, and \Cref{sec:learningdyn} extends this to Hamiltonian evolution. \Cref{sec:discussion} discusses limitations and heuristic extensions. We finish the paper with a brief conclusion section in \Cref{sec:conclusion}.

\section{Background}
\label{sec:background}
\subsection{Probably Approximately Correct learning}
PAC learning theory provides a formal mathematical framework for learning tasks. PAC analyzes so-called concept classes $\mathcal{C}$, which are sets of functions called concepts $c_\alpha$:
$$\mathcal{C} = \{c_\alpha : x \in \mathcal{X} \rightarrow y = c_{\alpha}(x) \in \mathcal{Y}\}_{\alpha \in \mathcal{A}}\,.$$
The key goal in PAC learning is to approximate an unknown concept from a specified class given access to a dataset of examples labeled by the concept. 
That is, for a fixed unknown concept $c_{\alpha_0}$ in the concept class, we are given $T$ input-output pairs. 
The inputs $x_t$ are sampled from a probability distribution $D$ over the input domain $\mathcal{X}$, and the outputs are $c_{\alpha_0}(x_t)$. 
We call this dataset $\mathcal{S}$.
\begin{equation}
    \mathcal{S} \coloneqq \{(x_t,y_t = c_{\alpha_0}(x_t))\}_{t\in[1,T]},~ x_t \overset{\mathrm{i.i.d.}}{\sim} D \,.
\end{equation}
Using such a dataset, the goal of a learning algorithm is to output a \textit{hypothesis} function $h:\mathcal{X}\rightarrow \mathcal{Y}$ which one can apply to unseen data $x_{new}$. In successful learning, the assigned labels $\hat{y}_{\text{new}} = h(x_{\text{new}})$ will be $\epsilon$-close to the true label with high probability $1-\delta$. 
A tuple $(\mathcal{C},D)$ is PAC learnable if there exists an algorithm which produces a hypothesis $h$ given only polynomially many examples, that is $T$ scales at most polynomially with $\delta^{-1}$, $\epsilon^{-1}$ and the  size of the input.  
Adding to PAC learnability a notion of computational costs, a concept class $\mathcal{C}=\{\mathcal{C}_n\}_{n}$\footnote{Here the concept class is divided into sub-classes specific to input size $n$ to make the dependence more explicit.} increasing in size $n$, which yields the concept of \textbf{efficient} PAC learning, defined as follows.
\begin{definition}[Efficient PAC learnability]
\label{def:pac}
The concept class $\mathcal{C}=\{\mathcal{C}_n\}_n$ is \textit{efficiently PAC learnable} if there exists a $\poly(\epsilon^{-1},\delta^{-1},n)$-time algorithm $\mathcal{A}$ 
such that for all $\epsilon\geq 0 $, all $0\leq\delta\leq 1$, all $n$, 
for any target concept $c:\mathcal{X}_n\rightarrow\mathcal{Y}_n$ in $\mathcal{C}_n$ and any target distribution $\mathcal{D}_n$ on $\mathcal{X}_n$, if $\mathcal{A}$ receives in input a training dataset $\{(x_t,c(x_t))\}_{t\in[0,T]}$ of $\poly(\epsilon^{-1},\delta^{-1},n)$-size $T$, then with probability at least $1-\delta$ over the random datasets, the learning algorithm $\mathcal{A}$ outputs a specification of a hypothesis function $h=\mathcal{A}(T,\epsilon,\delta)$ running in $\poly(\epsilon^{-1},\delta^{-1},n)$-time in satisfies
\begin{equation}\label{eq:real pac}
    \operatorname{Pr}_{x \sim \mathcal{D}} (\lvert h(x)- c(x)\rvert \leq \epsilon) \geq 1-\delta.
\end{equation}
\end{definition}

Note that this definition of PAC-learnability is distribution-independent; it should be true for any input distribution.
However, sometimes PAC learnability is distribution dependent, in which case we say that $\{(\mathcal{C},\mathcal{D})\}$ is PAC learnable.

 We say that if the above is true for \textit{classical} polynomial time, then the concept class is \textit{classically efficiently learnable}, and if either $\mathcal{A}$ or $h$ or both run in \textit{quantum} polynomial time, then the concept class is \textit{quantum efficiently learnable}.
 
 In the following sections, we will be proving quantum learnability. 
We will also prove classical non-learnability, that is, the existence of a classical learning algorithm implying unlikely complexity class inclusions.

\subsection{Parametrized Quantum circuits}
Parametrized quantum circuits, and in particular variational methods \cite{cerezo_variational_2021}, have been at the centre of recent approaches to quantum machine learning and were extensively studied \cite{schuld_introduction_2015}. 
Here, we focus on parametrized quantum circuits where the inputs and other tunable parameters appear repeatedly as Pauli rotations. 
\begin{definition}[Pauli encoding]
\label{def:paulienc}
    A Pauli-encoded circuit is a parametrized quantum circuit on $n$ qubits $U: \alpha \in [0,1]^{d} \rightarrow \mathcal{U}(2^n)$. It is composed of $N_f \in \poly(n)$ fixed unitary gates and $L\in \poly(n)$ parametrized gates $\{V_l(\alpha) \coloneqq e^{i\pi P_l\alpha_{i_l}}\}_{1\leq l \leq L}$ where $P_l$ are Pauli strings.
\end{definition}
It has been shown that such Pauli Quantum Circuits admit a finite Fourier representation \cite{schuld_effect_2021}.
\begin{lemma}
\label{lem:fourier}
    Any circuit $U$ as defined in \Cref{def:paulienc} admits a finite Fourier representation as follow:
    \begin{equation}
        \ket{\phi(\alpha)} = U(\alpha) \ket{0} = \sum_{1 \leq k \leq 2^n } \sum_{ l \in [-L,+L]^d} a_{l,k} e^{i\pi \alpha\cdot l} \ket{k} \,.
    \end{equation}
\end{lemma} 
Measuring the expectation value of some observable $O$ for such a state results in what we call a \textit{PQC function}, an input-output mapping as follows,
\begin{equation}
    f(\alpha) = \bra{0}U^{\dagger}(\alpha)OU(\alpha)\ket{0}\,.
\end{equation}
In the case of quantum reuploading models with Pauli encodings as in \Cref{def:paulienc}, we have that
\begin{equation}
    f(\alpha) = \sum_{l\in[-2L,2L]^d} b_l e^{i\pi\alpha \cdot l} \label{eq:1}\,.
\end{equation}
We call the coefficients $b_l$ the \textit{Fourier coefficients} of the \textit{PQC function} $f$.

In the next section, we provide a subroutine for sampling from them.
Building on this subroutine, we provide a sampling-based algorithm for their estimation.

\section{\subroutine~algorithm}
\label{sec:freqsamp}
\subsection{Fourier representation of parametrized circuits}
For our quantum learning algorithm, we introduce a new subroutine that allows us to prepare a state that amplitude-encodes the coefficients $b_l$ of a PQC function $f$, and describe a sampling-based method for their extraction.
As we explain in \Cref{app:oracles}, it is possible to estimate the coefficients given just appropriate black-box access to the classical function $f$.
However, here we assume access to the gate-decomposition of the circuit, which yields a more elegant and significantly more gate-frugal method.
Based on this gate decomposition, we propose an algorithm that transforms the description of this circuit to the description of a circuit with an additional register for frequencies.
\begin{theorem}
    There exists an algorithm $\mathcal{A}$ that given the description of any parametrized circuit $U$ as defined in \Cref{def:paulienc}, returns the description of
    a non-parametrized quantum circuit $U'=\mathcal{A}(U)$ on $n'$ qubits with $L'$ gates, such that it prepares the Fourier representation state of $U$ as follows:
    \begin{equation}
        \ket{\phi'} = \mathcal{A}(U) \ket{0} = \sum_{1 \leq k \leq 2^n } \sum_{ l \in [-L,+L]^d} a_{l,k} \ket{l} \ket{k} \,.
    \end{equation}
    with $n' = n+d\lceil{\log{L}\rceil}+1$ qubits and $L' = N_f + L(2n+d\lceil{\log{L}\rceil})$ gates.
    \label{thm:fourqru}
\end{theorem}
The detailed algorithm and the proof for this theorem are provided in the \Cref{app:fourqru}. 
We provide a high-level explanation of this algorithm.
The fixed gates (the ones without data uploading) are left unchanged by the algorithm $\mathcal{A}$. 
For the reuploading gates, we note that any data Pauli-uploading gates can be transformed into a $Z$ string by adding appropriate basis change gates.
The algorithm $\mathcal{A}$ transforms a data-encoding gate with a $Z$ Pauli on the bitstring $x$ into an increment (decrement) gate on the frequency register controlled on the even (odd) parity of $x$. 
This is illustrated in \Cref{fig:freqsamp}.

\subsection{Fourier representation of expectation values}
As seen in \Cref{lem:fourier}, when the inputs $\alpha$ are Pauli-encoded, the PQC function has a finite Fourier decomposition. 
We can connect \Cref{thm:fourqru} with this representation to obtain a sampling algorithm for the coefficients $b_l$ as given in \Cref{eq:1}.
First, we construct a circuit that returns the Fourier decomposition of such a quantum function amplitude encoded on a quantum state.

The key to do so is to realise that as $f(\alpha) = \bra{0} U(\alpha)^\dagger PU(\alpha)\ket{0}$, applying the algorithm $\mathcal{A}$ to the state $U(\alpha)^\dagger PU(\alpha)\ket{0}$ and then post-selecting the all-zero state yields the Fourier decomposition of the quantum function $f$ as a quantum state. The full proof is in the \Cref{app:fourexp}.

\begin{corollary}
\label{cor:statefourfun}
    For every circuit $U$ as defined in \Cref{def:paulienc} and Pauli observable $P$, we define the quantum function $f$:
    \begin{equation}
       \alpha\in\mathbb{R}^d \stackrel{f}{\longrightarrow} \bra{0} U(\alpha)^{\dagger} P U(\alpha) \ket{0}\,.
    \end{equation}
    $f$ has a finite Fourier representation, there exists a vector $b \in \mathbb{C}^{d(4L+1)}$ indexed by  $l\in\mathcal{L}\coloneqq[-2L,+2L]^d$ such that
    \begin{equation}
        f(\alpha) = \sum_{ l \in \mathcal{L}} b_l e^{i\pi \alpha \cdot l}\,.
    \end{equation}
    Then there is a quantum algorithm $\mathcal{A}$ with complexity $O(\poly(n,d))$ that retrieves the state amplitude encoding of the Fourier coefficients of $f$ on the $\ket{0}$ subspace.
    \begin{equation}
        \mathcal{A}(U)\ket{0}\ket{0} = \frac{1}{\lVert b \rVert_2} \sum_{ l \in \mathcal{L}} b_l \ket{l}\ket{0}+\ket{\cdots}\ket{0^\perp}\,.
    \end{equation}
\end{corollary}
Note that the probabilistic component here is unavoidable, as in general there is no reason for the function $f$ to have a Fourier decomposition which is a unit vector. 
Note that the theorem above is specific to observables that are Pauli strings $P$. In the \Cref{app:fourexp}, we prove that this result can efficiently be extended to a broader range of observables $O$, such as linear combinations of Pauli terms with a polynomial number of terms and polynomial spectral norm and local projectors.

Using $\mathcal{A}$ as a subroutine, we can perform \textit{\subroutine}, i.e., extract any coefficient $b_l$ up to additive error, by using controlled versions of $\mathcal{A}(U)$ in a Hadamard test. 
\begin{corollary}
    \label{cor:coefextract}
    With the same premise as in \Cref{cor:statefourfun}, there exists a $\poly(n,\epsilon^{-1})$ quantum algorithm that retrieves $b_l$ up to additive error $\epsilon$ for every $l$.
\end{corollary}
Furthermore, the Fourier amplitude-encoded state realizes a particular quantum feature map, which can be used to construct kernel-based machine learning methods even when the spectrum is exponentially large.

\begin{figure*}
\resizebox{\linewidth}{!}{
    \begin{quantikz}
        \lstick{$\ket{0}$} & \qw & \gate[3]{U_A} & \gate{Z(\pi \alpha_1)} 
        & \gate[3]{U_B} & \gate[3]{ZZZ(\pi \alpha_2)}& \\
        \lstick{$\ket{0}$} & \qw &               &                
        &               &                      & \\
        \lstick{$\ket{0}$} & \qw &               &                
        &               &                      &
    \end{quantikz}
    \begin{quantikz}
        \lstick{$\ket{0}$} & \qwbundle{m} &               & \gate{V_+} & \gate{V_-} &                 &          &          &          &            &            &          &          &          &\\
        \lstick{$\ket{0}$} & \qwbundle{m} &               &            &            &                 &          &          &          & \gate{V_+} & \gate{V_-} &          &          &          &\\
        \lstick{$\ket{0}$} & \qw          & \gate[3]{U_A} & \octrl{-2} &  \ctrl{-2} & \gate[3]{U_B}   & \ctrl{3} &          &          &            &            & \ctrl{3} &          &          &\\ 
        \lstick{$\ket{0}$} & \qw          &               &            &            &                 &          & \ctrl{2} &          &            &            &          & \ctrl{2} &          &\\  
        \lstick{$\ket{0}$} & \qw          &               &            &            &                 &          &          & \ctrl{1} &            &            &          &          & \ctrl{1} &\\  
        \lstick{$\ket{0}$} & \qw          &               &            &            &                 & \targ{}  & \targ{}  & \targ{}  & \octrl{-4} & \ctrl{-4}  & \targ{}  & \targ{}  & \targ{}  &  
    \end{quantikz}
}
\caption{(left) A parametrized circuit $U(\alpha)$ yielding a state $\sum_k \sum_l a_{l,k} e^{i\pi l\cdot \alpha} \ket{k}$. (right) The corresponding circuit $\mathcal{A}(U)$ returning a state amplitude-encoding Fourier coefficients as  $\sum_k \sum_l a_{l,k} \ket{l} \ket{k}$}.
\label{fig:freqsamp}.
\end{figure*}
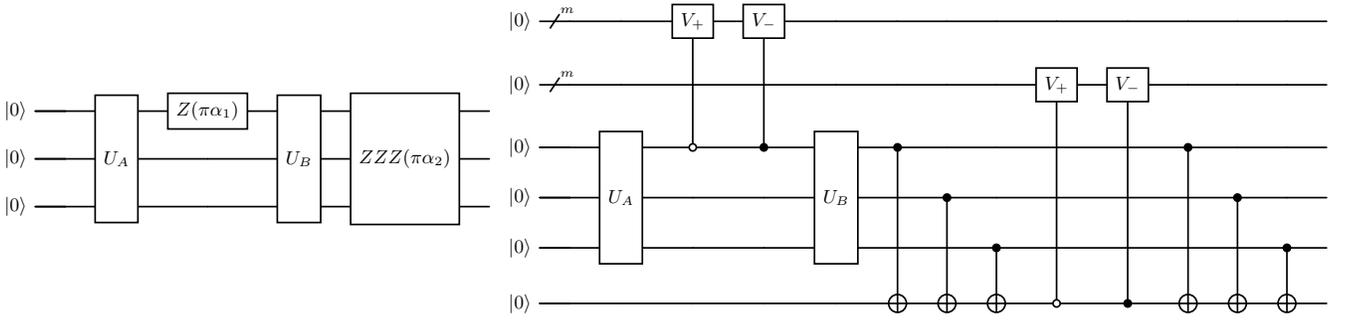

\section{Learning Parametrized Quantum Circuits}
\label{sec:learningpqc}
In this section, we study a concept class based on parametrized quantum circuits, which we call the \cca~ concept class. 
Then we describe an efficient quantum learner, based on the \textit{\subroutine} algorithm proposed in the previous section. 
Finally, we show that this concept class is hard to learn classically, effectively proving a learning separation. 

\subsection{Concept Class Definition}
We consider a family of parametrized circuits with Pauli encodings, as in \Cref{def:paulienc}. We define $x$ as the bitstring describing the fixed gates of the circuit. 
The parametrized gates have a known nature ($P_l$) but are parametrized by unknown parameters ($\alpha_l$). 
This gives rise to the following concept class, which we call \textit{\cca}. 
\begin{definition}[\cca~concept class]
\label{def:parcircCC}
    Consider a parametrized circuit $U$ on $n$ qubits as defined in \Cref{def:paulienc}, we note as $x$ the bit string describing the fixed gates, and $\alpha \in [0,1]^{d_n}$ the input parameters of the circuit. We have:
    \begin{equation}
        U(x,\alpha) \ket{0} = \sum_{1 \leq k \leq 2^n } \sum_{ l \in [-L,+L]^d_n} a_{l,k}(x) e^{i\pi l\cdot \alpha} \ket{k} \,.
    \end{equation}
    For a given observable $O$, we define the concept class $\mathcal{C}^{PQC}_f=\{\mathcal{C}^{PQC}_{n,f}\}_n$, for a qubit number $n$ and a number of unknown parameters $d_n$ scaling with $n$ as $$d_n = f(n)$$ as follows,
    \begin{multline}
        \mathcal{C}^{PQC}_{n,f} \coloneqq \{c_\alpha : x\in\{0,1\}^{n} \rightarrow\\
        \bra{0} U(x,\alpha)^{\dagger} O U(x,\alpha) \ket{0} \}_{\alpha \in [0,1]^{f(n)}}\,.
    \end{multline}
\end{definition}

This concept class has been studied, in particular, its covering number has been formally upper-bounded \cite{caro_generalization_2022}, yielding a good generalization performance. 
As a consequence, as long as $d\in\poly(n)$, this class is PAC-learnable. 
Here we will be focusing on the possibility and impossibility of \textit{efficient} PAC learning.

\subsection{Efficient Quantum Learner}
\label{subsec:circlearn}
We consider the concept class as defined in \Cref{def:parcircCC}.
In this subsection, we aim to prove that $\mathcal{C}^{PQC}_{\log}$ is quantum efficiently PAC learnable as in \Cref{def:pac}.
To do so, we present an efficient quantum learning algorithm that uses the Fourier sampling procedure described in \Cref{sec:freqsamp}. 

We are given a $T$ sized dataset $\{x_t,y_t\}_{1\leq t \leq T}$ for an unknown fixed concept $c_\alpha$ for data sampled over $\mathcal{D}$, a maximum probability of failure as $\delta$ and the required accuracy as $\epsilon$.
The problem size is the number of qubits $n$, and we fix the number of parameters as $d=\log(n)$, each uploaded a constant number of times $L$. 
This setup yields a number of frequencies $m=|\mathcal{L}| = (4L+1)^d\in \Theta(\poly(n))$. 
We describe the training and inference stage of the proposed algorithm.
 
\textbf{Training:} First we choose the measurement accuracy $\epsilon_b$ for the Fourier coefficients. 
For each Fourier coefficient, a simple sample average strategy is used, requiring $\epsilon^{-2}_b$ shots (a quadratic improvement may be possible here). 
For each bitstring $x_t$ specifying a circuit and for each frequency $l$, the coefficient $b_l(x_t)$ is retrieved using \Cref{cor:coefextract} up to additive error $\epsilon_b$. 
This procedure requires the execution of $mT/\epsilon_b^2$ poly-depth quantum circuits.
The retrieved Fourier coefficients are stacked in the matrix $\hat{B}$ where the rows correspond to the frequencies indexed by $l$ and the columns to the samples indexed by $t$. 
Similarly, we define the exact Fourier coefficients as $B$ and the shot noise $E=B-\hat{B}$.

We have that $y=b(x)\cdot w(\alpha)$, where $w$ is an unknown weight vector with $w_l(\alpha) \coloneqq e^{i\pi l\cdot\alpha}$. 
In other words, the concept is linear with respect to $B$, and stacking all labels in a vector $Y$, we have $Bw=Y$. 
Our goal is to minimize the mean square error on unseen data. 
We look for a hypothesis function $h$ that minimizes the error defined as.
\begin{equation}
    \mathcal{L}_c(h) = \mathbb{E}_{x\sim D}[\abs{h(x)-c_{\alpha}(x)}^2]\,.
\end{equation}
We use a regression technique (LASSO, see details in \Cref{app:lasso}) that yields a mean square error minimizer, with an additional constraint on the 1-norm of the weight vectors, which should be upper bounded by $\Lambda_1$.
Minimizers returned by this regression technique belong to the following set of functions, which we define as our hypothesis class:
\begin{equation}
    \mathcal{H}=\{b \xrightarrow[]{h_w} b\cdot w \mid \norm{w}_1 \leq \Lambda_1\}
\end{equation}
Using \Cref{thm:lasso} with $\Lambda_1 = \norm{w}_1 = m$ and $r_\infty=1$ we know that in order to reach $\epsilon$ accuracy on the mean square error with $1- \delta$ accuracy we can choose the parameters as follows, (noting that $\epsilon_y = 0$).
\begin{align}
T = \frac{16 m^4 \sqrt{ 2 \log\left( \frac{2m}{\delta} \right)}}{\epsilon^2},~
\epsilon_b = \frac{0.2 \epsilon}{m}
\end{align}
We note that $T\in \tilde{\Theta}(\poly(n,\epsilon^{-1},\delta^{-1}))$. In addition, the training stage of our algorithm outputs a weight vector $\hat{w}$ executing of $\tilde{\Theta}(m^7/\epsilon^4) \subset \tilde{\Theta}(\poly(n,\epsilon^{-1},\delta^{-1}))$ poly-depth quantum circuits. Therefore, both the number of samples and the runtime required to output a model are polynomial in $n$, $\epsilon^{-1}$, and $\delta^{-1}$.

\textbf{Inference:} Given a new datapoint $x_{t'}$, one retrieves the Fourier coefficients $b_{t'}$ using the algorithm in \Cref{sec:freqsamp} execution of $m/\epsilon_b^2 \in \tilde{\Theta}(m^3/\epsilon^2)  \subset \tilde{\Theta}(\poly(n,\epsilon^{-1},\delta^{-1}))$ poly-depth circuits. 
Then using the weight vector $\hat{w}$ derived in the training phase, the model returns $\hat{y}_{t'} = \hat{w} \cdot \hat{b}(x_{t'})$. 
\Cref{thm:lasso} guarantees that $\lvert \hat{y}_{t'} - y_t \rvert \leq \epsilon$ with probability $1-\delta$.
The runtime of the model is polynomial in $n$, $\epsilon^{-1}$ and $\delta^{-1}$. 

This concludes the proof of our first main result, which we summarize in the following theorem.
\begin{theorem}
\label{thm:PQCqeff}
    The \cca~concept class $\mathcal{C}^{PQC}_{\log}$ as in \Cref{def:parcircCC} is efficiently quantum PAC learnable as in \Cref{def:pac}.
\end{theorem}

\subsection{Hardness of the learning problem}
\label{subsec:hardcirc}
In this subsection, we prove that the \cca~ concept class is not classically efficient PAC learnable unless $\mathsf{BQP} \subset \mathsf{P/poly}$, for more on this class see \Cref{app:clxty}. 

 We choose a promise $\mathsf{BQP}$ language $\mathcal{L}$, solved by $U \in \mathcal{U}(2^n)$, meaning that the sign of $\Tr[O U \ketbra{x}{x} U^\dagger]$ decides $\mathcal{L}$, where $O$ is some simple observable. We define the concept class as in \Cref{def:parcircCC} as the circuit $U$ with added parametrized gates $P_l$ at fixed locations with unknown parameters $\alpha_l$. 

The concept class contains at least one concept that is BQP-hard, for $\alpha=0$. 
Therefore, using Lemma 2 in \cite{molteni_exponential_2024} no classical learner can learn $c_0$ (if \Cref{conj:1} is true). 
This yields the following theorem.
\begin{theorem}
\label{thm:PQChard}
If $\mathsf{BQP} \not\subset \mathsf{P/poly}$ then the \cca~ concept class \Cref{def:parcircCC} is not classically efficient PAC learnable as in \Cref{def:pac}.
\end{theorem}

\Cref{thm:PQChard} and \Cref{thm:PQCqeff} together prove a separation between quantum and classical learners for the \cca~ problem. 
However, this case is rather contrived and does not have any direct physical application. 
In the next section, we adapt this result to when a circuit compiles a time evolution, which yields a more practical learning problem.

\section{Learning Time Evolution}
\label{sec:learningdyn}
\subsection{Concept Class Definition}
Our main result is a learning separation for the Hamiltonian dynamics learning problem. 

We consider a parametrized Hamiltonian $H(x,\alpha)$ the input is $x \in  \mathbb{R}^n$ and $\alpha \in \mathbb{R}^d$ are unknown parameters. 
We are interested in its time evolution for a fixed time $\tau$, as $U_n(x,\alpha) = e^{i\tau(H(x,\alpha))}$.

This gives rise to the following concept class, which we call \textit{\ccb}.
\begin{definition}[\ccb~concept class]
\label{def:CCte}
    Consider a sequence of parametrized Hamiltonian $\{H_n(x,\alpha)\}_n$, where each $H_n$ operates on $n$ qubits and is described by continuous parameters $\alpha \in [0,1]^d_n$ and bitstrings $x$ of length $s_n$. 
    For a given sequence of observables $\{O_n\}_n$, we define the concept class $\mathcal{C}^H_f=\{\mathcal{C}^H_{f,n}\}_n$, for $n$ qubits, $d_n$ parameters scaling with $n$ as $d_n=f(n)$, and a fixed real number $\tau$ , resulting in a time evolution $U_n(x,\alpha) = e^{i\tau H_n(x,\alpha)}$, as follows,
    \begin{multline}
        \mathcal{C}^H_{n,f} \coloneqq \{c_\alpha : x\in\{0,1\}^{s_n} \rightarrow\\ \bra{1} U_n(x,\alpha)^{\dagger} O U_n(x,\alpha) \ket{0} \}_{\alpha \in [0,1]^{f(n)}} \,.
    \end{multline}
\end{definition}

\subsection{Connection between the two concept classes}

The \ccb~concept class can be related to the \cca~concept class via Hamiltonian simulation.
Given the class $\mathcal{C}^{{H}}_{n,d_n}$, by using Hamiltonian simulation on the underlying Hamiltonians (taking into account the parametrizations), we obtain parametrized circuits. 
The precision we use in Hamiltonian simulation dictates how closely functions in one class approximate the functions in the other in a precisely quantified way, and at the same time, the depth of the parametrized circuit. 

We illustrate this with an example. 
Consider the Ising Hamiltonian on an arbitrary graph with a transverse field of unknown strength $\alpha$.
The arbitrary graph is described by bitstrings indicating whether an edge exists or not, $x_{i,j} =1$ if $(i,j)\in E$.
\begin{equation}
    H(x,\alpha) = \sum_{i,j} x_{i,j} Z_i Z_j + \alpha \sum_i  X_i
\end{equation}
The first order Trotterization of $r$ steps, yields the following parametrized quantum circuit for the parameter $\alpha\tau/r$.
\begin{equation}
    U_r(x,\alpha) = \left(\prod_{i,j} Z_iZ_j(x_{i,j}\tau/r) \prod_{i} X_i(\alpha\tau/r)\right)^r
\end{equation}
This is illustrated in \Cref{fig:hampqc}.

We use this connection to show that the learning algorithm that learns $\mathcal{C}^{{PQC}}_{\log}$ can also learn the time-dynamics class $\mathcal{C}^{{H}}_{\log}$. 
In the rest of this section, we will adapt the results of the previous section to a quantum circuit approximating this evolution. 
As the quantum learner we devised previously is only efficient for polynomially sized spectrum, we investigate circuit compilation such that this is guaranteed.   

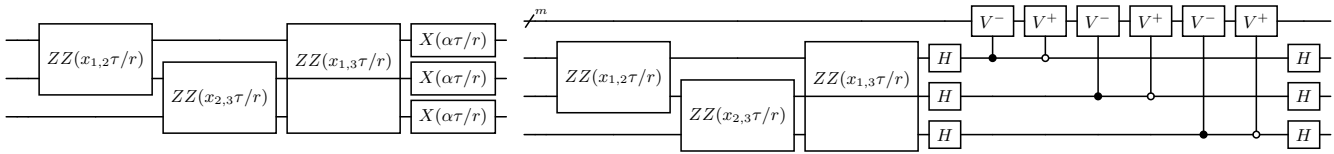
\begin{figure*}
    \resizebox{\linewidth}{!}{
        \centering
        \begin{quantikz}[row sep=0.1cm, column sep=0.2cm,transparent]
            % %freq register
            % \lstick{$\ket{0}$} & \qwbundle{m} && 
            % &
            % &% end ZZ
            % & 
            % \\
            % qubit 1
            % \lstick{$\ket{0}$} 
            &              && \gate[2]{ZZ(x_{1,2}\tau/r)}                           
            & 
            &\gate[3,label style={yshift=0.3cm}]{ZZ(x_{1,3}\tau/r)} % end ZZ
            & \gate{X(\alpha\tau/r)}&
            \\
            % qubit 2
            % \lstick{$\ket{0}$} 
            &              &&
            & \gate[2]{ZZ(x_{2,3}\tau/r)}
            &\linethrough % end ZZ
            & \gate{X(\alpha\tau/r)}&
            \\
            % qubit 3
            % \lstick{$\ket{0}$} 
            &              &&
            &
            &% end ZZ
            & \gate{X(\alpha\tau/r)}&
        \end{quantikz}
        \begin{quantikz}[row sep=0.1cm, column sep=0.2cm,transparent]
            %freq register
            % \lstick{$\ket{0}$} 
            & \qwbundle{m} &&&&% end ZZ
            && \gate{V^-} &\gate{V^+}
            & \gate{V^-} &\gate{V^+}
            & \gate{V^-} &\gate{V^+} &
            &
            \\
            % qubit 1
            % \lstick{$\ket{0}$} 
            &              && \gate[2]{ZZ(x_{1,2}\tau/r)}                           
            & 
            &\gate[3,label style={yshift=0.3cm}]{ZZ(x_{1,3}\tau/r)} % end ZZ
            &\gate{H}& \ctrl{-1} & \octrl{-1} && &&
            & \gate{H}
            &
            \\
            % qubit 2
            % \lstick{$\ket{0}$} 
            &              &&
            & \gate[2]{ZZ(x_{2,3}\tau/r)}
            &\linethrough % end ZZ
            & \gate{H}&& & \ctrl{-2} & \octrl{-2} && 
            &\gate{H}
            &
            \\
            % qubit 3
            % \lstick{$\ket{0}$} 
            &              &&
            &
            &% end ZZ
            & \gate{H} && && & \ctrl{-3} & \octrl{-3}
            & \gate{H}
            &
        \end{quantikz}
    }
    \caption{Example of the Trotterization of the $\tau$-time evolution of $H(x,\alpha) = \sum_{i,j} x_{i,j} Z_i Z_j + \alpha \sum_i  X_i$ for three qubits in $r$ step. The circuit presented is just one Trotter step and shall be repeated $r$ times. (left) Compilation as a Pauli encoded parametrized quantum circuit as in \Cref{def:paulienc}. (right) \subroutine~applied to the original circuit.}
    \label{fig:hampqc}
\end{figure*}

% \subsection{Efficient Quantum Learner}
With the same strategy as that of the previous section, we show that the learning algorithm for the \cca~concept class also works as an efficient learning algorithm for the \ccb~concept class.
For completeness, we have derived the full proof in \Cref{app:learnccb} for first order Trotterization.
The key idea is that Hamiltonian simulation allows us to construct a good hypothesis class for our learning algorithm based on the concept class. In fact, as we explain later, all Hamiltonian simulation methods lead to function families with essentially the same learning characteristics.
\begin{theorem}
\label{thm:TEqeff}
    The \ccb~concept class $\mathcal{C}^{H}_{\log}$ as in \Cref{def:CCte} is efficiently quantum PAC learnable as in \Cref{def:pac}.
\end{theorem}

\subsection{Hardness of the learning problem}
In this subsection, we prove that the \ccb~ concept class is not classically efficient PAC learnable unless $\mathsf{BQP} \subset \mathsf{P/poly}$.

Specifically, we will be proving that the special case of the \ccb~ concept class, where the input is a bit string that specifies the initial state, and the Hamiltonian is parametrized only by $\alpha$ is hard.

For this, we use Lemma 3 in \cite{molteni_exponential_2024}, based on \cite{childs_quantum_2004}, which is restated here.
\begin{lemma}
    \label{lem:Hhard}[from \cite{molteni_exponential_2024}]
    For any $k$-gate quantum circuit $U = U_k\cdots U_2U_1$ acting on $n$ qubits there exists a local Hamiltonian $H$ such that for any $n$ qubit initial state $\ket{\psi}$, we have 
    \begin{equation}
        e^{iH\pi}\ket{\psi}\ket{0} = U \ket{\psi}\ket{k}
    \end{equation}
\end{lemma}

We choose a promise $\mathsf{BQP}$ language $\mathcal{L}$, solved by $U \in \mathcal{U}(2^n)$, meaning that the sign of $\Tr[O U \ketbra{x}{x} U^\dagger]$ decides $\mathcal{L}$, where $O$ is some simple observable. We define the $V(x)$ that prepares $\ket{x}$ from $\ket{0}$. 
We apply \Cref{lem:Hhard} and get $H(x)$ such that 
\begin{equation}
    e^{iH(x)\pi}\ket{0}\ket{0} = U V(x) \ket{0}\ket{k} = U\ket{x}\ket{k}
\end{equation}
Choosing the observable $O'=O\otimes I$, we define the concept class as in \Cref{def:CCte} for the hamiltonian $H(x,\alpha)$ with added parametrized Pauli terms $P_l$ with unknown parameters $\alpha_l$. 
\begin{equation}
    H(x,\alpha) = H(x) + \sum \alpha_l P_l
\end{equation}

The concept class contains at least one concept that is BQP-hard, for $\alpha=0$. 
Therefore, using Lemma 2 in \cite{molteni_exponential_2024}, no classical learner can learn $c_0$ (under specified complexity assumptions). 
This yields the following theorem.
\begin{theorem}
\label{thm:PQChard}
If $\mathsf{BQP} \not\subset \mathsf{P/poly}$ then the \ccb~ concept class \Cref{def:CCte} is not classically efficiently PAC learnable as in \Cref{def:pac}.
\end{theorem}

\section{Beyond log-many parameters}

\subsection{Exponentially large spectrum}

So far we have considered concept classes where the number $d_n$ of $\alpha$ parameters scales logarithmically with $n$, making them restricted.
When instead of $d_n$ scales polynomially, the cardinality of the spectrum $\lvert\mathcal{L}\rvert$ scales exponentially.
The proposed algorithm is no longer efficient in general, as it would require performing regression in an exponentially large space.
In addition, for complexity-theoretic reasons, devising a scheme that provably efficiently PAC-learns any setting with an exponentially large spectrum is not possible.

Specifically, in \cite{arunachalam_quantum_2019} it was proven that the learning of shallow classical circuits, even in the quantum PAC model (strictly stronger model than ours as the data in terms of a purification rather than samples from a distribution) is impossible under common complexity-theoretic assumptions, e.g. that ring-learning with errors cannot be done in polynomial time on a quantum computer.
Polynomially-sized PQCs with polynomially many parameters can encode shallow classical circuits, and thus their efficient learnability would also imply the unlikely efficient quantum algorithms for learning with errors.

However, giving up on provable bounds, we can still devise a heuristic algorithm making use of the proposed feature map.

\subsection{Kernel approach}
We propose a kernel method approach as a heuristic.
Consider the \cca~concept class, for which we have $c_{\alpha}(x) = \sum_{l\in\mathcal{L}}b_le^{i\pi\alpha\cdot l}$, or equivalently, for any $l\in\mathcal{L} = [-L,L]^d$:
\begin{equation}
    b_l(x) = \int_{(0,1)^d} c_{\alpha}(x) e^{-i\pi \alpha \cdot l} d\alpha
\end{equation}
We define the kernel as,
\begin{equation}
    k(x,x') = b(x) \cdot b(x')\,,
\end{equation}
and make use of the quantum algorithm proposed in \Cref{sec:freqsamp} to estimate kernel values.
We provide more details on the circuit we use to do so efficiently in \Cref{app:kernel}.
Based on this, we build the $T\times T$ Gram matrix with $O(T^2)$ evaluations on a quantum computer. 
Finally, we can perform traditional kernel methods on a classical quantum computer, for example, kernel ridge regression.

We discuss the limitations of this kernel approach and describe conditions under which they can be mitigated.
In general, the dimension of the feature map of the resulting kernel is exponentially large.
Therefore, the generalization performance is not guaranteed unless we have an exponential number of training samples.
In fact, quantum kernels famously suffer from problems of exponential concentration \cite{thanasilp_exponential_2024}. 
A number of known factors yield exponential concentration, such as the expressivity of the data embedding, or global measurements. 
However, looking at the circuit producing the kernel evaluation in \Cref{fig:algokernel}, we argue that the kernel we propose is not immediately concerned by any of the known causes of exponential concentration in quantum kernels.

In fact, we prove that if the spectrum is sparse, with polynomially large support, the kernel ridge regression yields an efficient PAC learning algorithm in \Cref{app:kernel}.
It is possible to construct an artificial case where we get a frequency support that is only polynomially large for an a priori exponentially large spectrum.
Consider the concept class \Cref{def:parcircCC} with $d\in O(\poly(n))$, but such that most $\alpha$ parameters cancel each other out by construction, with for example $R_z(\alpha_s)YR_z(\alpha_s)$.
Suppose that only logarithmically many $\alpha$ survive these cancellations.
While the spectrum is a priori exponentially large, it is in reality only polynomially large.
The proposed kernel approach could efficiently learn this concept class, while no classical learner could unless $\mathsf{BQP} \subset \mathsf{P/poly}$.

\subsection{Properties of the feature map}
By Mercer's theorem, any kernel has a feature map associated with it. In the kernel approach that we propose, the feature map is simply $x\rightarrow b(x)$.
That is, the feature map is directly the complex vector that represents the Fourier coefficients with respect to $\alpha$ of the function described by a bitstring $x$.

In this subsection, we argue that for the \ccb~concept class, the feature map associated with the kernel is in some sense equivalent for any circuit approximating the function. 
In particular, this means that as long as a given precision level is reached, the feature map is invariant with respect to the  Hamiltonian time evolution technique.
For simplicity, we propose to use Trotterization, where the depth scales polynomially with the error, but, for more optimal schemes \cite{low_hamiltonian_2019, wiebe_hamiltonian_2012}, the depth scales logarithmically with the error. 
We explain this in more detail below.

Given a bounded function $f:[0,1]^d \rightarrow \mathbb{C}$, it is always possible to define its Fourier coefficients as $b_l = \int_{(0,1)^d} f(\alpha) e^{i2\pi l \cdot \alpha} d\alpha$ for any $l\in\mathbb{Z}^d$. 
That is, any such function can be represented as a complex vector in an infinite-dimensional space.
We have also seen that when a parametrized circuit is Pauli-encoded, it has a finite Fourier representation, that is, $b$ is a finite-dimensional vector. 
We argue that if a Pauli-encoded parametrized quantum circuit approximates a function, its finite-dimensional vector approximates the infinite-dimensional one in the 2-norm.
The Hausdorff–Young inequality \cite{christ_sharpened_2014} states that, for any two integer $p\in[1,2]$ and $p'$ such that $\frac{1}{p}+\frac{1}{p'}=1$, we have,
\begin{equation}
    \Big(\sum_{l\in\mathbb{Z}^d}\big|b_l\big|^{p'}\Big)^{1/p'}\leq\Big(\int_{(0,1)^l}|f(\alpha)|^p\,d\alpha\Big)^{1/p}\,.
\end{equation}

Consider an exact time evolution yielding a concept $c_\alpha$ (see \Cref{def:CCte}).
It is approximated up to $\epsilon/2$ by a first circuit yielding the functions $c'$ and another circuit yielding the function $c''$.
We note their Fourier decomposition $b'$ and $b''$ respectively. 
Then, using the Hausdorff–Young inequality with $p=p'=2$ we get an upper bound of the 2-norm of the Fourier decomposition of the difference $c'-c''$ as its infinite norm, itself bounded by $\epsilon$, as follows,
\begin{equation}
    \lVert b'_l - b''_l \rVert_2\leq\Big(\int_{(0,1)^d}|(c'-c'')(\alpha)|^2\,d\alpha\Big)^{1/2} \leq \epsilon\,.
\end{equation}
This concludes that feature maps approximating the same quantum function are $\epsilon$-close in the 2-norm metric space.
This analysis implies that the learning and generalization properties of the kernel will not in any significant way depend on which Hamiltonian simulation technique is used, i.e. which hypothesis function is chosen for the quantum learning algorithm.

\section{Discussion}
\label{sec:discussion}
\subsection{Cardinality of the Concept class}

In the first work demonstrating quantum-classical learning separations for quantum functions the concept classes were polynomially sized, which meant a brute-force algorithm checking all hypotheses is efficient \cite{gyurik_exponential_2023}. 
In \cite{molteni_exponential_2024,yamasaki_advantage_2023} first examples of learning with an exponentially-sized concept class were introduced. 

The concept classes we study here, the \cca~and the \ccb~concept classes, are indexed by continuous parameters, and each parameter setting mathematically specifies a different function (up to periodicity concerns), and thus the class is a continuum. 
However, we note that we are interested in approximations, that is, finding functions which agree with the true function on a $1-\delta$ fraction of the inputs when sampled from the input distribution. 
The question then arises of what the effective size of this function family is.
Precisely, we are interested in the hypothesis family $\mathcal{H}$, such that for each $c \in \mathcal{C}$ there exists $h_c \in H$ such that $h_c$ is $\epsilon$-close to $c$ in the PAC sense as in \Cref{def:pac}. $H$ can depend on $\epsilon$ and the input distribution $\mathcal{D}$.

If $\mathcal{H}$ is polynomially sized in $n$ and its elements can be efficiently constructed and evaluated, then this would allow for a brute-force type algorithm for our learning task in the log-parameter case
We note that demanding that $\mathcal{H}$ is a subset of $\mathcal{C}$ we obtain the notion of $\epsilon$-packing of the concept class, which is closely related to covering numbers, and both are quite well understood as quantities governing, for example, generalization bounds~\cite{caro_generalization_2022}.

For our concept class \cca, to the best of our effort to obtain a tight bound (see \Cref{app:card}) we can find a super-polynomial grid of functions which attains epsilon-packing, upper bounding the smallest effective hypothesis size to $O(n^{\log\log(n)})$. 
We conjecture that this log log scaling in the exponent may be an artefact of the bounding method and that, in fact, the concept classes with $d$ in $O(\log(n))$ allow for a polynomial-sized effective class and a more direct learning algorithm. 
However the key question of whether this is true and whether these hypotheses can be efficiently found and evaluated remains open.

Therefore, the method we provide is not brute force, but the arguments above raise the question of whether a brute force method could also be possible for our learning task.
Either way, we emphasize that a learning separation persists.

\subsection{Conclusions}
\label{sec:conclusion}
In this work, we have proposed an algorithm that efficiently prepares a state representing the Fourier decomposition of some quantum functions.
We have defined a concept class based on parametrized quantum circuits that can be efficiently PAC learned using a quantum computer and standard regression techniques.
As this first concept class is not physically relevant, we built on it and propose a concept class based on a Hamiltonian time evolution and show that it is also quantum-efficient PAC learnable.
We do so by applying the previous result to the Trotterized time evolution, but later discuss that any quantum simulation technique would yield similar performance.
We also show that both classes cannot be PAC-learned efficiently by any classical algorithm unless $\mathsf{BQP} \subset \mathsf{P/poly}$, effectively proving a learning separation.
Both concept classes have polynomially sized feature space, which yields a weaker learning separation in comparison with brute force approaches \cite{gyurik_exponential_2023}. 
We finally discuss regimes in which a priori exponentially large feature space could remain efficiently PAC-learnable and avoid exponential concentration issues \cite{thanasilp_exponential_2024}.

\section{Acknowledgments}
\label{sec:acknowledgments}
The authors are grateful to Matthias Caro and Sofiene Jerbi for numerous discussions.
AB and MG are supported by CERN through the CERN Quantum Technology Initiative. 
AB is supported by the Quantum Computing for Earth Observation (QC4EO) initiative of ESA $\phi$-lab, partially funded under contract 4000135723/21/I-DT-lr, in the FutureEO program. 
This work was supported by the Dutch National Growth Fund (NGF), as part of the Quantum Delta NL programme.
VD was supported by ERC CoG BeMAIQuantum (Grant No. 101124342). 
This work was also partially supported by the Dutch Research Council (NWO/OCW), as part of the Quantum Software Consortium programme (Project No. 024.003.03).
MYR is part of the project Divide \& Quantum (with project number 1389.20.241) of the research programme NWA-ORC which is (partly) financed by the Dutch Research Council (NWO)
Views and opinions expressed are however, those of the author(s) only and do not necessarily reflect those of the European Union, the European Commission, or the European Space Agency, and neither can they be held responsible for them.

\bibliographystyle{plain}
\bibliography{references.bib}

\appendix
\onecolumngrid
\section{LASSO regression}
\label{app:lasso} 
\label{subsec:lasso}
LASSO (Least Absolute Shrinkage and Selection Operator) regression is a linear regression method \cite{mohri_foundations_2018} with an added regularization on the 1-norm of the weight vector to the loss function. This enforces sparsity in the estimated coefficients. Given a input matrix $X \in \mathbb{R}^{m \times T}$ and label vector $y \in \mathbb{R}^{m}$, the LASSO estimator solves a mean square error minimization with a constraint of the norm 1 of the weight vector.
\[
\hat{w} = \arg\min_{w} \| y - Xw \|_2^2 \text{~subject to~}\|w\|_1\leq \Lambda_1,
\]  
where $\Lambda_1 > 0$ controls the trade-off between sparsity and model fit. This property makes LASSO particularly useful in high-dimensional settings where many predictors may be irrelevant. In \Cref{app:lasso} we prove the following theorem about the robustness of the LASSO regressor in the presence of perturbations.

\begin{theorem}
Consider a dataset of size \( T \), denoted \( \{(\hat{b}_t,\hat{y}_t)\}_{1 \leq t \leq T} \), generated from a linear model with an unknown weight vector \( w \in \mathbb{R}^m \), and affected by perturbations, as follows:
\begin{align}
    b_t \cdot w &= y_t, \\
    \hat{b}_t &= b_t + \eta_{b,t}, \\
    \hat{y}_t &= y_t + \eta_{y,t},
\end{align}
where \( \norm{\eta_{b,t}}_\infty \) and \( \norm{\eta_{y,t}}_\infty \) are bounded by \( \epsilon_b \) and \( \epsilon_y \), respectively. Here, \( m \) denotes the dimension of $b$ and we know $r_\infty$ such that $\norm{b}_\infty \leq r_{\infty}$. Running the LASSO algorithm with the constraint that $\norm{w}_1 \leq \Lambda_1$ on this dataset, with probability at least \( 1 - \delta \), yields a model \( h^* \) such that the true risk (training error + generalisation error) is at most \( \epsilon \), provided that
\begin{align}
T &\geq \frac{\left(2 \Lambda_1 r_\infty \right)^4 \sqrt{ 2 \log\left( \frac{2m}{\delta} \right)}}{\epsilon^2}.\\
\Lambda_1 \epsilon_b &\leq 0.2 \epsilon\\
\epsilon_y &\leq 0.5 \epsilon
\end{align}
\label{thm:lasso}
\end{theorem}

We prove \Cref{thm:lasso} below. First, we state the well-known generalization bound for the LASSO algorithm. Then we prove a lemma that bounds the empirical risk, and then we combine these two to bound the true risk.
\begin{theorem}
\label{thm:genlasso}
 Let \( \mathcal{X} \subseteq \mathbb{R}^m \) and 
\[
\mathcal{H} = \left\{ b \in \mathcal{X} \mapsto w \cdot b : \|w\|_1 \leq \Lambda_1 \right\}.
\]
Let \( S = ((b_1, y_1), \ldots, (b_T, y_T)) \in (\mathcal{X} \times \mathcal{Y})^T \). Let \( \mathcal{D} \) denote a distribution over \( \mathcal{X} \times \mathcal{Y} \) according to which the training data \( S \) is drawn. Assume that there exists \( r_\infty > 0 \) such that for all \( b \in \mathcal{X} \), \( \|b\|_\infty \leq r_\infty \), and \( M > 0 \) such that \( |h(b) - y| \leq M \) for all \( (b, y) \in \mathcal{X} \times \mathcal{Y} \). Then, for any \( \delta > 0 \), with probability at least \( 1 - \delta \), each of the following inequalities holds for all \( h \in \mathcal{H} \):
\[
\mathcal{R}(h) \leq \hat{\mathcal{R}}_S(h) + 2 r_\infty \Lambda_1 M \sqrt{\frac{2 \log(2m)}{T}} + M^2 \sqrt{\frac{\log(\delta^{-1})}{2T}},
\]
where \( \mathcal{R}(h) = \mathbb{E}_{(b, y) \sim \mathcal{D}} \left[ |h(b) - y|^2 \right] \) is the prediction error for the hypothesis \( h \), and \( \hat{\mathcal{R}}_S(h) \) is the training error of \( h \) on the training data \( S \). 
\end{theorem} 
\begin{lemma}[Upper Bound on Empirical Risk under Bounded Perturbations]
\label{lem:emprisk}
Let \( h^*(x) = \mathbf{w}^* \cdot \hat{\mathbf{b}}(x) \) denote the hypothesis returned by the LASSO algorithm, trained on a dataset with perturbed features \( \hat{\mathbf{b}}(x) \) and noisy labels \( \hat{y} = y + \eta_y \). Suppose the true labeling function is linear, i.e., \( y = \mathbf{w} \cdot \mathbf{b}(x) \), for some weight vector \( \mathbf{w} \in \mathbb{R}^m \), and assume that \( \|\mathbf{w}\|_1 \leq \Lambda_1 \). Furthermore, assume that the unperturbed features satisfy \( \|\mathbf{b}(x)\|_\infty \leq r_\infty \), the perturbation on the features is bounded by \( \|\hat{\mathbf{b}}(x) - \mathbf{b}(x)\|_\infty \leq \epsilon_b \), and the additive noise on the labels is bounded as \( |\eta_y| \leq \epsilon_y \). If the LASSO optimization problem is solved approximately such that the empirical risk of \( h^* \) is within \( \epsilon_3 / 2 \) of the optimal empirical risk over all weight vectors with \( \ell_1 \)-norm at most \( \Lambda_1 \), then the empirical risk of the resulting hypothesis satisfies
\[
\hat{\mathcal{R}}_S(h^*) = \frac{1}{T} \sum_{t=1}^N \left( h^*(x_t) - y_t \right)^2 \leq \left( \Lambda_1 \epsilon_b + \epsilon_y \right)^2 + \frac{\epsilon_3}{2}.
\]
\end{lemma}

\begin{proof}
For any function \( g \), the training error is defined as
\[
\hat{\mathcal{R}}(g) = \frac{1}{T} \sum_{t=1}^T |g(x_t) - y_t|^2.
\]
Now let \( \mathbf{w}^* \) be the optimal vector that the LASSO algorithm outputs, i.e.,
\[
\mathbf{w}^* = \arg\min_{\|\mathbf{w}\|_1 \leq \Lambda_1} \left( \frac{1}{T} \sum_{t=1}^T \left| \mathbf{w} \cdot \hat{\mathbf{b}}(x_t) - y_t \right|^2 \right).
\]
It is clear that for any other \( \mathbf{w}' \neq \mathbf{w}^* \) with bounded norm (i.e. $\norm{\mathbf{w}'}_1\leq \Lambda_1$),
\[
\frac{1}{N} \sum_{i=1}^N \left| \mathbf{w}^* \cdot \hat{\mathbf{b}}(x_i) - y_i \right|^2 \leq \frac{1}{N} \sum_{i=1}^N \left| \mathbf{w}' \cdot \hat{\mathbf{b}}(x_i) - y_i \right|^2.
\]
Let \( \mathbf{w} \) be the true weight vector in the true labeling function \( h(x) = \mathbf{w} \cdot \mathbf{b}(x) \). Using this, we can bound the training error for the function \( h'(x) = \mathbf{w}' \cdot \hat{\mathbf{b}}(x) \). Let
\[
t^* = \arg\max_{0 \leq t \leq T} |\mathbf{w} \cdot \hat{\mathbf{b}}(x_t) - y_t|^2
\]
be the index of the training data point that maximizes the loss. Then:
\begin{align*}
\hat{\mathcal{R}}(h') &= \frac{1}{T} \sum_{t=1}^T |\mathbf{w}' \cdot \hat{\mathbf{b}}(x_t) - y_t|^2 \\
&\leq \frac{1}{T} \sum_{t=1}^T |\mathbf{w} \cdot \hat{\mathbf{b}}(x_t) - y_t|^2 \\
&\leq |\mathbf{w} \cdot \hat{\mathbf{b}}(x_{t^*}) - y_{t^*}|^2 \\
&\leq \left( |\mathbf{w} \cdot \hat{\mathbf{b}}(x_{t^*}) - \mathbf{w} \cdot \mathbf{b}(x_{t^*})| + |\mathbf{w} \cdot \mathbf{b}(x_{t^*}) - y_{t^*}| \right)^2 \\
&\leq (\Lambda_1 \epsilon_b + \epsilon_y)^2.
\end{align*}
Now consider the function \( \hat{h}(x) = \mathbf{\hat{w}} \cdot \hat{\mathbf{b}}(x) \), where \( \mathbf{\hat{w}} \) is obtained by minimizing the training error such that its empirical risk is at most \( \epsilon_3 / 2 \) worse than the minimum. That is, we allow for a suboptimal solution. Then the empirical risk of \( \hat{h} \) satisfies:
\[
\hat{\mathcal{R}}(\hat{h}) \leq (\Lambda_1 \epsilon_b + \epsilon_y)^2 + \frac{\epsilon_3}{2}.
\]
\end{proof}
The proof of \Cref{thm:genlasso} follows from the prediction error of the LASSO algorithm and the bound on empirical risk derived in \Cref{lem:emprisk}. First, we address \(M\), which is the upper bound on the absolute prediction error. In the noisy case where LASSO receives \(\hat{b}\) as input and the labels are also noisy, we have:
\[
|\hat{h}(\hat{b}) - \hat{y}| \leq M.
\]
Now:
\begin{align*}
|\hat{h}(\hat{b}) - \hat{y}| &= |\hat{b} \cdot \hat{w} - \hat{y}| \\
&= |\hat{b} \cdot \hat{w} - y - \eta_y| \\
&\leq |\hat{b} \cdot \hat{w}| + |y + \eta_y| \\
&\leq \|\hat{b}\|_\infty \|w\|_1 + |y| + |\eta_y|,
\end{align*}
where the last inequality follows from Hölder's inequality. The error on the labels is bounded, so \( |\eta_y| \leq \epsilon_y \), and \( \|w\|_1 \leq \Lambda_1 \). Since the true label is \( y = b \cdot w \), and under the assumptions \( \|b\|_\infty \leq r_\infty \) and \( \|w\|_1 \leq \Lambda_1 \), we get \( |y| \leq \|b\|_\infty \|w\|_1 \leq r_\infty \Lambda_1 \). Similarly, since \( \|\hat{b}\|_\infty \leq r_\infty + \epsilon_b \), and assuming \( \|\hat{w}\|_1 \leq \Lambda_1 \), we obtain the bound:
\[
M := \Lambda_1(2r_\infty + \epsilon_b) + \epsilon_y.
\]
Using this lemma, we can rewrite the generalization bound as:
\[
\mathcal{R}(\hat{h}) \leq (\Lambda_1 \epsilon_b + \epsilon_y)^2 + \frac{\epsilon_3}{2} + 2 r_\infty \Lambda_1 M \sqrt{\frac{2 \log(2m)}{T}} + M^2 \sqrt{\frac{\log(\delta^{-1})}{2T}},
\]
where \( r_\infty = r_\infty + \epsilon_b \). To bound the prediction error above by
\[
\epsilon = (\Lambda_1 \epsilon_b + \epsilon_y)^2 + \epsilon_3,
\]
it suffices to choose \( T \) such that:
\[
2 r_\infty \Lambda_1 M \sqrt{\frac{2 \log(2m)}{T}} + M^2 \sqrt{\frac{\log(\delta^{-1})}{2T}} \leq \frac{\epsilon_3}{2},
\]
and substituting for \( M \) and \( r_\infty \), solving for \( T \), gives the sample complexity:
\[
T \geq \frac{\left( \Lambda_1(2r_\infty + \epsilon_b) + \epsilon_y \right)^4 \sqrt{2 \log(2m / \delta)}}{\epsilon_3^2}.
\]
By setting \( \Lambda_1 \epsilon_b = 0.2 \epsilon \), \( \epsilon_y = 0.5 \epsilon \), and \( \epsilon_3 = 0.4 \epsilon \), the prediction error is bounded by
\[
(\Lambda_1 \epsilon_b + \epsilon_y)^2 + \epsilon_3 \leq \epsilon,
\]
provided that
\[
T > \frac{(2 \Lambda_1 r_\infty)^4 \sqrt{2 \log(2m / \delta)}}{\epsilon^2}.
\]
which concludes the proof of \Cref{thm:lasso}.

\section{Complexity assumption}
\label{app:clxty}
 
 In this work, we present learning separation statements that rely on the following widely believed conjecture. 

 \begin{conjecture}
 \label{conj:1}
    $\mathsf{BQP}\not\subset \mathsf{P}/\mathsf{poly}$.
\end{conjecture}
This is a weaker conjecture than the following one used in previous works.
 \begin{conjecture}
 \label{conj:2}
    There exists a distribution $\mathcal{D}$ such that $\mathsf{BQP}\not\subset \mathsf{HeurP}^{\mathcal{D}}\mathsf{/poly}$.
\end{conjecture}

Nonetheless, this too is believed to be true: the discrete logarithm problem or factoring are not believed to be in $\mathsf{HeurP}\slash \mathsf{poly}$. For these problems, the average-case to worst-case reductions imply that if either problem is efficiently solvable heuristically, then it is also efficiently solvable in the worst case as well. 
This does not hold for  $\mathsf{BQP}$ functions in general, so Conjecture 1 is indeed weaker. Still, in general, by Lemma 3 in \cite{gyurik_exponential_2023} if there exists a single $\mathcal{L} \in \mathsf{BQP}$ that is not in $\mathsf{HeurP/poly}$ under some distribution, then for every $\mathsf{BQP}$-complete problem, there exists a distribution under which the problem is not in $\mathsf{HeurP/poly}$. 
\begin{lemma}[from \cite{gyurik_exponential_2023}]
If there exists a $(L, \mathcal{D}) \not\in \mathsf{HeurP/poly}$ with $L \in \mathsf{BQP}$, then for every $L' \in \mathsf{BQP}$-$\mathsf{complete}$ there exists a family of distributions $\mathcal{D}' = \{\mathcal{D}'_n\}_{n \in \mathbb{N}}$ such that $(L', \mathcal{D}') \not\in \mathsf{HeurP/poly}$.
\label{lem:heurppolybqp}
\end{lemma}

\section{\subroutine~algorithm}
\label{app:fourqru}
\subsection{Fourier representation of parametrized circuits}
In this section, we prove the \Cref{thm:fourqru} and describe the algorithm $\mathcal{A}$. As mentioned in the main text, without loss of generality, we consider that every circuit as defined in \Cref{def:paulienc} uses strings of identity and $Z$ matrices encoding. Indeed, if encoding with $X$ or $Y$ appear they can be changed to $Z$ with local unitary gates, which can be absorbed in the set of fixed gates. 

\subsubsection{Description of the algorithm}
The input to the algorithm $\mathcal{A}$ is a description of the $n$-qubit circuits as an alternating sequence between sets of fixed gates resulting in unitaries $\{U_l\}_{0\leq l \leq L}$ and encoding gates $\{e^{i \alpha_{s_l} \prod_{1\leq k \leq n} Z_k^{b_{l,k}}}\}_{1\leq l \leq L}$ where $s_l$ describe the index of the dimension of the data that is encoded, and $b_l$ is the bitstring of the qubits affected by the encoding gate. 

The output of the algorithm $\mathcal{A}$ is a description of a circuit on a larger number of qubits. Registers are added to the existing circuit register to keep track of frequencies of $\alpha$ on each of the dimensions, with each register requiring a number of qubits depending on the number of times that this feature is reuploaded, as follows:
\begin{equation}
    n_s = 1+\lceil2\log(\sum_l s_l=s)\rceil
\end{equation}

There is also a single additional ancillary qubit that is used to compute parities; therefore, in total, there are $n_T$ qubits as follows:
\begin{equation}
    n_T = n + 1 + \sum n_s
\end{equation}

We name the registers $f= f_0 \cdots f_s$ for the frequency registers, $a$ for the ancillary qubit, and $c$ for the circuit register. A gate $G$ applied to the register $r$ will be written as $G^{(r)}$. Control on the $\ket{1}$ ($\ket{0}$) state are written as $C$ ($\bar{C}$). The frequency registers are also indexed by negative numbers, we write the basis as follows $\{\ket{k}\}_{-(n_s-1)/2 \leq k \leq (n_s-1)/2}$
For the frequency register we define the increment (decrement) gate with unitary matrix $V_+$ ($V_-$=$V_+^{\dagger}$) ensuring unitarity with a circular condition as such:
\begin{align}
    V_+\ket{k} &=\ket{(k+1)}, \forall k<(n_s-1)/2 \\
    V_+\ket{(n_s-1)/2}&=\ket{-(n_s-1)/2}
\end{align}

The algorithm $\mathcal{A}$ returns a sequence of gates that follows that of the original circuit, where fixed gates are unchanged and applied to the $c$ register, and encoding gates $\{s_l,b_l\}$ are transformed as follows. For each qubit affected non trivially by the encoding gate (that is, with a $Z_k$ generator when $b_{l,k}=1$) the parity is computed on the ancillary qubit, with a sequence of CNOT gates as follows:
\begin{equation}
    D(b_l) = \prod_{k\mid b_{l,k}=1} C^{(c_k)}X^{(a)}
\end{equation}
At this stage, the ancillary encodes the parity of the sub-bitstring for indexes $b_l$. Then the frequency register of the dimension of the input being encoded $s_l$ is acted on based on the parity encoded in the ancillary as such:
\begin{equation}
    G(s_l) = \left(C^{(a)}V_+^{(f_{s_l})}\right) \left(\bar{C}^{(a)}V_-^{(f_{s_l})} \right)
\end{equation}
Effectively, the subspace where the parity of the substring is even sees an increment in the frequency of the input being encoded, while the other subspace sees a decrement. 
Finally, the ancillary qubit is reset to be reused later with $D(b_l)^\dagger = D(b_l)$. 

The output of the algorithm is an alternating sequence between sets of gates being unchanged from fixed gates resulting in unitaries $\{U_l^{(c)}\}_{0\leq l \leq L}$ and gates replacing encoding gates as $\{D(b_l)G(s_l)D(b_l)\}_{0\leq l \leq L}$.

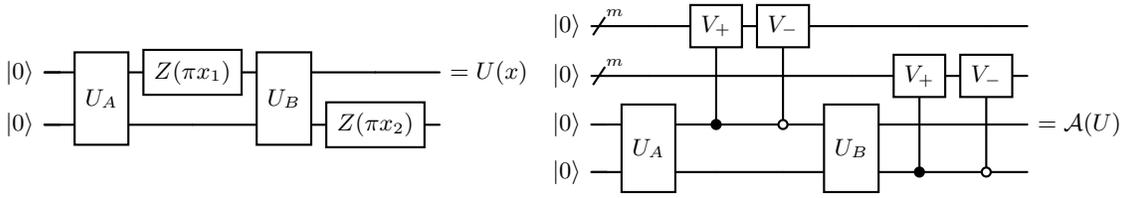
\begin{figure}[h]
    \centering
    \begin{quantikz}[row sep=0.1cm, column sep=0.2cm]
        \lstick{$\ket{0}$} & \qw & \gate[2]{U_A} & \gate{Z(\pi x_1)} 
        & \gate[2]{U_B} &                 & \rstick{$=U(x)$}\\
        \lstick{$\ket{0}$} & \qw &               &                
        &               & \gate{Z(\pi x_2)} &
    \end{quantikz}
    \begin{quantikz}[row sep=0.1cm, column sep=0.2cm]
        \lstick{$\ket{0}$} & \qwbundle{m} &               & \gate{V_+} 
        & \gate{V_-} &                 &            &           &\\
        \lstick{$\ket{0}$} & \qwbundle{m} &               &
        &            &                 & \gate{V_+} & \gate{V_-}&\\
        \lstick{$\ket{0}$} & \qw          & \gate[2]{U_A} & \ctrl{-2}   
        & \octrl{-2} & \gate[2]{U_B}   &            &           & 
        \rstick{$=\mathcal{A}(U)$}\\
        \lstick{$\ket{0}$} & \qw          &               &             
        &            &                 & \ctrl{-2}  & \octrl{-2} &  
    \end{quantikz}
    \caption{Illustration of the \textit{\subroutine} algorithm for multiple inputs}
    \label{fig:example-algomulti}
\end{figure}

\subsubsection{Proof of the algorithm}
Finally, we prove that the algorithm $\mathcal{A}$ indeed outputs what it is promised to return. That is, given $U$ such that
\begin{equation}
        \ket{\phi} = U\ket{0} = \sum_{1 \leq k \leq 2^n } \sum_{ l \in [-L,+L]^d} a_{l,k} e^{i\alpha\cdot l}\ket{k} \,,
\end{equation}
then it returns, 
\begin{equation}
    \ket{\phi'} = \mathcal{A}(U) \ket{0} = \sum_{1 \leq k \leq 2^n } \sum_{ l \in [-L,+L]^d} a_{l,k} \ket{l}^{(f)} \ket{0}^{(a)}\ket{k}^{(c)} \,.
\end{equation}
For the rest of the proof, we shorten the indices $k$ and $l$ for easier notation. This is an induction proof in three steps:
\begin{enumerate}
    \item Initial step : This is trivial, if $\ket{\phi} = \ket{0}$ then $\ket{\phi'}$  = $\mathcal{A}(I) \ket{0} = \ket{0}^{(f)}\otimes \ket{0}^{(a)} \otimes \ket{0}^{(c)}$.
    \item Fixed gate step: Suppose that the algorithm works as intended, before the application of a set of fixed gates with unitary $U$. Let us show that as algorithm $\mathcal{A}$ applies $U' = I^{(f)}\otimes I^{(a)} \otimes U^{(c)}$, it yields the correct state. It is easy to see that $U\ket{\phi} =  \sum_{l} (\sum_{k}  a_{l,k} U \ket{k}) e^{i\alpha\cdot l}$, and $U'\ket{\phi} =  \sum_{l} (\sum_{k}  a_{l,k} U \ket{k}^{(c)}) \ket{0}^{(a)} \ket{l}^{(f)}$.
    \item Encoding gate step: Suppose that the algorithm works as intended before the application of an encoding gate $\{s_l,b_l\}$. Let us prove that applying $D(b_l)G(s_l)D(b_l)$ yields the correct state. We prove this below.
\end{enumerate}
If these three steps are true, then by induction, the algorithm is correct.

\textbf{Proof of step 3.} The effect of the encoding gate $\{s_l,b_l\}$ on the state $\ket{\phi}$ is as follows
\begin{equation}
    e^{i \alpha_{s_l} \prod_{k'} Z_{k'}^{b_{l,k'}}} \ket{\phi} =  \sum_{k} \sum_{l}  a_{l,k} e^{i \alpha_{s_l} p(k,b_l)} \ket{k} e^{i\alpha\cdot l} = \sum_{k} \sum_{l}  a_{l,k}  \ket{k} e^{i\alpha\cdot (l + e_{s_l} p(k,b_l))}
\end{equation}
Where $p(k,b_l)$ is the parity of the substring of $k$ with indices $b_l$ as follows $p(k,b_l)= (-1)^{\sum_{k'}k_{b_{l,k'}}}$. We also write $e_{r}$ as the vector such that only the $r$-th element is $1$ and the rest is $0$.

On the other hand, for the $\mathcal{A}(U)$ computation we have:
\begin{align}
    D(b_l) \ket{\psi'} &=  
    \sum_{k} \sum_{l} a_{l,k} \ket{l}^{(f)} (X^{\sum_{k'} k_{b_{k'}}}\ket{0}^{(a)}) \ket{k}^{(c)}\\
    G(s_l) D(b_l) \ket{\psi'} &=
    \sum_{k} \sum_{l} a_{l,k} \ket{l + e_{s_l} p(k,b_l)}^{(f)} (X^{\sum_{k'} k_{b_{k'}}}\ket{0}^{(a)}) \ket{k}^{(c)}\\
    D(b_l) G(s_l) D(b_l) \ket{\psi'} &=
    \sum_{k} \sum_{l} a_{l,k} \ket{l + e_{s_l} p(k,b_l)}^{(f)} \ket{0}^{(a)} \ket{k}^{(c)}\\
\end{align}
The states correspond to each other, which concludes the proof to the iterative step 3. 

\subsection{Frequency representation of expectation values}
\label{app:fourexp}
So far, we have seen how to get the Fourier representation of a quantum reuploading circuit, but what about quantum function, that is, the expectation values of some observable? 
In this section, we prove \Cref{cor:statefourfun} and \Cref{cor:coefextract}.
We explain how to do this for Pauli observables and extend this result to linear combinations of Pauli observables and projectors.

\subsubsection{Pauli observable}

\begin{figure}[h]
    \centering
    \begin{quantikz}[row sep=0.1cm, column sep=0.2cm]
        \lstick{$\ket{0}$} & \gate[3]{\mathcal{A}(U)}&           
        &\gate[3]{\mathcal{A}(U)^{\dagger}}&&\\
        \lstick{$\ket{0}$} &                         &\gate[2]{P}
        &                                  &\gate[2]{\ketbra{0}{0}} \rstick{$=\mathcal{A}(U,P)$}\\
        \lstick{$\ket{0}$} &                         &           
        &                                  &
    \end{quantikz}
    \caption{Illustration of the quantum function evaluation algorithm}
    \label{fig:algo}
\end{figure}
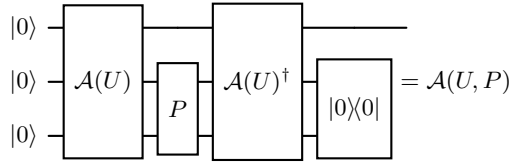

Without loss of generality (using local changes of basis), the Pauli string measurement can be considered a $Z$ string and therefore $\ket{\phi(\alpha)}$ has the following form (where $p(P,k)$ is a parity function).
\begin{equation}
    f(\alpha) = \langle P \rangle(\alpha) =  \sum_{-L\leq l',l \leq L} \sum_k (-1)^{p(P,k)} a_{l,k} a_{l',k}^* e^{i(l-l')\cdot\alpha} = \sum_{-2L \leq l \leq 2L} b_l e^{i l \cdot \alpha}
\end{equation}
We define the state $\ket{\phi(\alpha),P} \coloneq U(\alpha)^{\dagger}PU(\alpha)\ket{0}$. It has the following amplitude for $\ket{0}$.
\begin{equation}
    \braket{0}{\phi(\alpha),P}=\bra{0}U(\alpha)^{\dagger}PU(\alpha)\ket{0} = \sum_{-2K \leq k \leq 2K} b_k e^{i k \alpha}
\end{equation}
Using the algorithm $\mathcal{A}$ on the state $\ket{\phi(\alpha),P}$ we can loose the dependence on $\alpha$ and efficiently get the Fourier representation state:
\begin{equation}
    \ket{\phi,P} = \sum_{-2K \leq k \leq 2K} b_k \ket{k} \ket{0} + \text{trash}\,.
\end{equation}

Post-selecting on $\ket{0}$ has a success probability $\sum \abs{b_k}^2$.
The full algorithm is illustrated in \Cref{fig:algo}.

\subsubsection{Linear combination of Pauli observables}
In this subsection, we extend the procedure above to a more generic observable, more specifically, a linear combination of polynomially many Pauli strings. This is can be done using a Linear Combination of Unitaries approach. Supposing that the observable of interest may be decomposed as follows
\begin{equation}
    O = \sum_h \beta_h P_h
\end{equation}
Following the procedure described previously, one may prepare the states $\ket{\phi,P_h}$ for each $P_h$. By the linearity of the expectation value, adding them yields the expectation value of the full observable.
\begin{equation}
    \ket{\phi,O} = \sum_h \beta_h \ket{\phi,P_h}
\end{equation}
The addition may be done using a linear combination of unitary approach, which requires an additional register with a number of qubits logarithmic in the number of terms in the observable, and the preparation of the following state. 
\begin{equation}
    V_\beta \ket{0} = \frac{1}{\lVert\beta\rVert}\sum_h \beta_h \ket{h}
\end{equation}
It also requires post-selection, which causes overhead in complexity (polynomial for reasonable observables). Once the state is prepared, the same sampling and post-processing may be implemented.

\begin{figure}[h]
    \centering
    \begin{quantikz}[row sep=0.1cm, column sep=0.2cm]
        \lstick{$\ket{0}$} &              & \qw & \gate[2]{V_\alpha}
        &\ctrl{1}     &\octrl{1}    &\ctrl{1}     &\octrl{1}    
        &\gate[2]{V_\alpha^{\dagger}}        &\gate[2]{\ketbra{0}{0}}\\
        \lstick{$\ket{0}$} &              & \qw &                 
        &\ctrl{3}     &\ctrl{3}     &\octrl{3}    &\octrl{3}    
        &                                    &                       \\
        \lstick{$\ket{0}$} & \qwbundle{m} & \qw & \gate[3]{c(U)}
        &             &             &             &             
        &\gate[3]{\mathcal{A}(U)^{\dagger}}&   &
        \rstick{$=\mathcal{A}(U,\sum \alpha_h P_h)$}\\
        \lstick{$\ket{0}$} & \qwbundle{m} & \qw &                         
        &             &             &             &
        &                                    &    &\\
        \lstick{$\ket{0}$} & \qwbundle{n} & \qw &                   
        &\gate{P_1}   &\gate{P_2}   &\gate{P_3}   &\gate{P_4}   
        &                                    &\gate{\ketbra{0}{0}}   
    \end{quantikz}
    \caption{Illustration of the quantum function evaluation algorithm for arbitrary observables}
    \label{fig:example-algo2}
\end{figure}
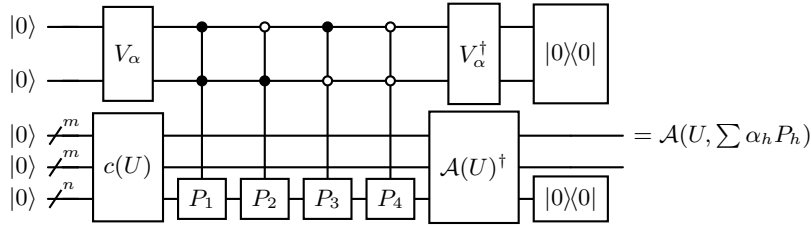

\subsubsection{Probabilities, or projectors observables}
Consider a circuit $U$ yielding a pure state $U\ket{0} = \ket{\phi} = \sum_k \phi_k \ket{k}$. Its density matrix is $\rho=\sum_{k,l} \phi_k \phi_l^* \ketbra{k}{l}$. Defining the conjugate circuit as $U^*$, we have
\begin{equation}
    (U \otimes U^* )(\ket{0} \otimes \ket{0})=\sum_{k,l} \phi_k \phi_l^* \ket{k}\otimes \ket{l}
\end{equation}
Therefore, if one wanted to retrieve the probability of measuring $\ket{0}$ of circuit $U$ one could retrieve the amplitude as above. This yields the procedure illustrated in \Cref{fig:example-projector} to retrieve the coefficients for the observable $\ketbra{0}{0}$, that is the probability of measuring $\ket{0}$ on the original circuit.
\begin{figure}[h]
    \centering
    \begin{quantikz}[row sep=0.1cm, column sep=0.2cm,transparent]
        \lstick{$\ket{0}$} & \qwbundle{m} &\qw & \gate[2]{\mathcal{A}(U)}
        &\gate[3,label style={yshift=0.3cm}]{\mathcal{A}(U^*)}&\qw&\\
        \lstick{$\ket{0}$} & \qwbundle{n} &\qw &                        
        &\linethrough                                         &\gate[2]{\ketbra{0}{0}}\\
        \lstick{$\ket{0}$} & \qwbundle{n} &\qw &                        
        &                                                     &
    \end{quantikz}
    \caption{Illustration of the quantum function evaluation algorithm for a projector}
    \label{fig:example-projector}
\end{figure}
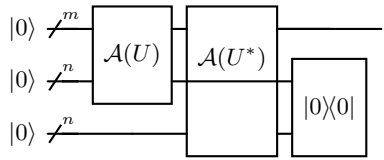

\subsubsection{Extracting Fourier coefficients}
Once the Fourier decomposition is available, one may be interested in retrieving the value coefficients. Suppose we have the following state as an output.
\begin{equation}
    \mathcal{A}(U,O) \ket{0} \ket{0} = \ket{\phi,O} = \sum_{-2K \leq k \leq 2K} b_k \ket{k} \ket{0} + \text{trash}
\end{equation}
Suppose one is interested in the coefficient of frequency $l$. We call $V_l$ any unitary such that $\ket{l}=V_l\ket{0}$, we have
\begin{equation}
    b_l = \bra{0}\bra{0} ( V_l^{\dagger} \otimes I) \mathcal{A}(U) \ket{0} \ket{0}
\end{equation}
We use a Hadamard test to extract the real and imaginary part of this coefficient as in \Cref{fig:example-algofour}. In addition, because the function is real, we have $b_k = b_{-k}$. Therefore, in case there is a polynomial number of frequencies, one may retrieve all of them efficiently, and create a classical surrogate for the quantum function. 
\begin{figure}[h]
    \centering
    \begin{quantikz}[row sep=0.1cm, column sep=0.2cm]
        \lstick{$\ket{0}$} & \gate{H}    & \ctrl{1}                
        & \ctrl{1} & \gate{H} & \meter{Z}\\
        \lstick{$\ket{0}$} & \qwbundle{m}& \gate[2]{\mathcal{A}(U)}
        &\gate{V_l}&          &\\
        \lstick{$\ket{0}$} & \qwbundle{n}&                         
        &          &          &
    \end{quantikz}
    \caption{Illustration of the \textit{\subroutine} algorithm}
    \label{fig:example-algofour}
\end{figure}
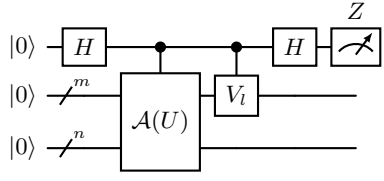

\section{Cardinality of the concept class}
\label{app:card}
First we will find a relatively easy bound for the covering number of the following set of functions:
\begin{equation}
    \mathcal{C} = \{c_{\alpha}: x \rightarrow \sum_{-L\leq l \leq L} b_l(x)e^{il\cdot\alpha},\lVert c_{\alpha}\rVert_{\infty}\leq1\}_{\alpha \in [0,1]^d }
\end{equation}
First we construct a covering net using the smoothness of the function with respect to its index $\alpha$. 
The partial derivative with respect to a single parameter is as follows,
\begin{equation}
     \lvert\partial_{\alpha_k} c_{\alpha}(x) \rvert\leq  \lvert  \sum_li l_k b_l(x)e^{il\cdot\alpha} \rvert \leq L \sum_l   \lvert b_l(x)e^{il\cdot\alpha} \rvert  \leq L \lVert b(x)\rVert_1\,.
\end{equation}
We can use it to bound the 1-norm of the gradient as,
\begin{equation}
    \lVert\nabla_\alpha c_{\alpha}(x)\rVert_{1} \leq d  L \lVert b(x)\rVert_1 \,.
\end{equation}
We consider a regular grid with $M$ divisions for each of the $d$ dimensions.
We call this grid $\mathcal{A} = \{\alpha_g\}_{g}$ , with $\lvert\mathcal{A}\rvert=M^d$. 
Therefore, each $\alpha$ will be at least $\lVert \delta_{\alpha}\rVert \coloneqq 1/M$ close to a grid point. 
We have 
\begin{align}
    \lVert c_{\alpha}-c_{\alpha_g} \rVert_{\infty} &\leq \nabla_\alpha c_\alpha \cdot (\alpha-\alpha_g) \\
    &\leq \lVert\nabla_\alpha c_{\alpha_g}\rVert_{1} \lVert\delta_\alpha(\alpha) \rVert_{\infty} \\
    &\leq dL\lVert b(x)\rVert_1 /M\,.
\end{align}

To guarantee the error is lower than $\epsilon$ for all $\alpha$ it is sufficient to choose
\begin{equation}
    M = dL \lVert b(x)\rVert_1/\epsilon\,.
\end{equation}
We are interested in the scaling of the grid when $d\in O(\log(n))$ and $L\in O(1)$. 
For all $x$ we have $\lVert b(x) \rVert_2 \leq 1$ and therefore $\lVert b(x)\rVert_1 \leq (2L+1)^{d/2} \in O(\poly(n))$.
This yields a number of grid points scaling as
\begin{equation}
    G=(dL\lVert b(x)\rVert_1/\epsilon)^d\in O(n^{\log(n)})\,.
\end{equation}

If we leave this function family $\mathcal{C}$, we can apply the approach of \cite{caro_generalization_2022} states bounds on covering numbers of shot-based expectation values of parametrized quantum circuits.
This allow us to get a smaller hypothesis class of proxy functions approximating the functions of the concept class.
Theorem 3 in the supplementary material of \cite{caro_generalization_2022} upper bounds the covering number for a PQC where $d$ parameters ($T$ in the original paper) are reuploaded $L$ times ($M$ in the original paper).
In order to get the expectation values up to $\eta$ additive accuracy, the number of reuploadings is multiplied by $\eta^{-2}$
Using $d\in O(\log n)$ and $L \in O(1)$ we get the following scaling:
\begin{equation}
    G \in O\left(\left(\frac{d L}{\epsilon \eta^2}\right)^d\right) \subseteq O(n^{\log\log n})\,.
\end{equation}

Although this scaling is closer to a polynomial bound, it is not polynomial, but this might be an artefact of the proof.
The key question of whether there exist a polynomial and constructable set of of hypothesis functions $\mathcal{H}$ that $\epsilon$-covers $\mathcal{C}$ remains open.

\section{Alternative Oracle-based algorithms}
\label{app:oracles}
Considering the settings of \Cref{sec:freqsamp}, with a $U(\alpha)$ defined as above, suppose that instead of the specification of the parametrized quantum circuit $U$ as a sequence of gates, we are given an oracle $O_U$ such that for a $m$-binary decomposition over the $d$-dimensional parameter $\alpha$, we have 
$$\ket{\alpha}\otimes\ket{\psi}\xrightarrow{O_U}\ket{\alpha}\otimes U(\alpha)\ket{\psi}\,.$$
Applying this oracle to the equal superposition state over the frequency registers, followed by a Quantum Fourier Transform applied to each coordinate frequency register, yields the same result as the algorithm described previously. We define a regular grid over the inputs $\{\alpha_l\}$ such that $\frac{1}{\sqrt{\lvert \mathcal{L} \rvert}} \sum_{l\in\mathcal{L}} \alpha_l = \ket{+}$
\begin{align}
    \ket{0}\ket{0}&\xrightarrow{H^{\otimes d m}\otimes I}
    \sum_{l\in\mathcal{L}}\ket{\alpha_l}\otimes\ket{0} \\
    &\xrightarrow{O_U} \sum_{l\in\mathcal{L}}\ket{\alpha_l}\otimes U(\alpha_l)\ket{0}\\
    &\xrightarrow{\operatorname{QFT}_m^{\otimes d}\otimes I}\sum_{l\in\mathcal{L}}\sum_{k\in[1,n]} a_{l,k} \ket{l}\ket{k}\,.
\end{align}
The above result provides an analogue for states decomposition as in \Cref{thm:fourqru}; similarly, we can obtain the analogue for PQC-functions as in \Cref{cor:statefourfun}.
For functions, an amplitude oracle can be defined as follows. 
For a function $f:\{0,1\}^*\rightarrow[0,1]$ that takes a binary decomposition over $\alpha$, an amplitude oracle $O_f$ yields:
$$\ket{\alpha}\otimes\ket{0}\xrightarrow{O_f}\ket{\alpha}\otimes \left(f(\alpha)\ket{0}+\sqrt{1-f(\alpha)^2}\ket{1}\right)\,.$$
As in the previous case, we initialize the $\alpha$ register in the equal superposition state, or equivalently in a superposition on the regular grid. Then we apply the oracle and then the QFT on the frequency registers to finally retrieve the Fourier decomposition of $f$ as an amplitude-encoded state.

\section{Proof of the PAC learnability of the \ccb concept class}
\label{app:learnccb}
\textbf{Training:} We consider quantum circuits realized by using $r$ Trotter steps of the time evolution of the  following parametrized Hamiltonian 
\begin{equation}
    H(x,\alpha)= H'(x) + \sum_{1\leq s \leq d}\alpha_s P_s\,,
\end{equation}
where $P_s$ are Pauli strings.
Such circuits have a frequency space of $m=|\mathcal{L}|=(4r+1)^d$. We implement the circuit learning algorithm exactly as in \Cref{subsec:circlearn}. The difference is that now the labels are off by the Trotter error, that is
\begin{equation}
    \epsilon_y < \frac{t^2 A}{2 r}\,,
\end{equation}
where $A$ is the sum of the spectral norm of all pairs of commutators. Using \Cref{thm:lasso} we are guaranteed that the labelling error on a new data point will be smaller $\epsilon$ with probability $1-\delta$ if we choose the following learning parameters
\begin{align}
T > \frac{16 m^4 \sqrt{ 2 \log\left( \frac{2m}{\delta} \right)}}{\epsilon^2},~
\epsilon_b < \frac{0.2 \epsilon}{m},~
\epsilon_y < 0.5 \epsilon\,.
\end{align}
The condition on $\epsilon_y$ yields a requirement on the number of Trotter steps as $r > t^2 A /\epsilon$ which in turn yields the feature space dimension as $m>(4t^2 A /\epsilon+1)^d$. We therefore require the number of training data points to scale as 
\begin{equation}
    T \in \tilde{\Theta}\left(\frac{ \sqrt{ d }}{\epsilon^{4d+2}}\right)\subset \tilde{\Theta}(\poly(n,\epsilon^{-1},\delta^{-1}))\,,
\end{equation}
so this method is sample-efficient. Regarding computational complexity, the training state requires the execution of $K = mT/\epsilon_b^2$ poly-depth quantum circuits with the condition of $\epsilon_b = 0.2\epsilon/m$. Recalling that $d\in O(\log(n))$,
\begin{equation}
    K\in \tilde{\Theta}\left(\frac{ \sqrt{ d }}{\epsilon^{7d+4}}\right) \subset \tilde{\Theta}(\poly(n,\epsilon^{-1},\delta^{-1}))\,.
\end{equation}
Therefore, the overall training process is efficient on a quantum computer. 

\textbf{Inference:} 
The inference process takes place analogously to the one in \Cref{subsec:circlearn}. Given a new datapoint $x_{t'}$, one retrieves the Fourier coefficients $b_{t'}$ of the approximate Trotter function. This requires the execution of $m/\epsilon_b^2 \in \tilde{\Theta}(m^3/\epsilon^2)  \subset \tilde{\Theta}(\poly(n,\epsilon^{-1},\delta^{-1}))$ poly-depth circuits. Then using the weight vector $\hat{w}$ derived in the training phase, the model returns $y_{t'} = \hat{w} \cdot \hat{b}(x_{t'})$.

\section{PAC efficient Kernel-based algorithm}
\label{app:kernel}
In this section, we propose a kernel based approach and proves that it efficiently PAC learns the parametrized circuit concept class as in \Cref{def:parcircCC} if the spectrum has a polynomially large spectrum.

\subsection{Concept class and noisy data}
Consider a $U_{\text{hard}}$ that decides a BQP-complete language with the function $c_0$ defined as follows
\begin{equation}
    c_0(x) \coloneqq \bra{x} U_{\text{hard}}^{\dagger} Z_0 U_{\text{hard}} \ket{x}
\end{equation}
We define the concept class where gates parametrized by an unknown vector $\alpha \in \mathbb{R}^d$ are added to the circuit implementing $U_{\text{hard}}$, such that the concept have a finite Fourier decomposition, as follows
\begin{equation}
    c_{\alpha}(x) = \langle b(x)\mid w(\alpha)\rangle, w(\alpha) = [e^{i \alpha \cdot l}]_{l\in[-L,+L]^d}, b(x),w(\alpha) \in \mathbb{C}^{m}\,.
\end{equation}
We have access to a dataset $\{(x_t,y_t=c_\alpha(x_t)\}_{t\in[0,T]}$. We define the feature space representation of the data $B=[b(x_t)]\in \mathbb{C}^{T\times m}$, and the Gram matrix $K=[\langle b(x_t)\mid b(x_{t'})\rangle]_{t,t'} \in \mathbb{C}^{T\times T}$. We have access to an approximation of the Gram matrix with $\hat{K} = K + E$ with $E$ p.s.d. and each element is bounded by $\epsilon_k$, which depends on the number of shots we use to measure the overlap, see \Cref{fig:algokernel}.

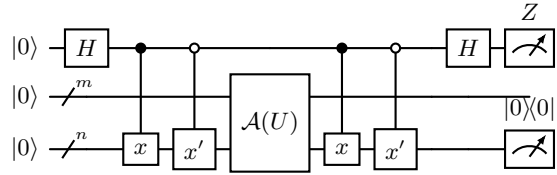
\begin{figure}[h]
    \centering
    \begin{quantikz}[row sep=0.1cm, column sep=0.2cm]
        \lstick{$\ket{0}$} & \gate{H}    & \ctrl{2} & \octrl{2} &&
        \ctrl{2} & \octrl{2} &
        &\gate{H}  & \meter{Z}\\
        \lstick{$\ket{0}$} & \qwbundle{m}&          &           &
        \gate[2]{\mathcal{A}(U)}&&&
        &          &          \\
        \lstick{$\ket{0}$} & \qwbundle{n}&\gate{x}  &\gate{x'}  &&
        \gate{x}  &\gate{x'}  
        &         && \meter{\ketbra{0}{0}}
    \end{quantikz}
    \caption{Evaluation of the kernel overlap, we have that $\langle Z\otimes I \otimes \ketbra{0}{0} \rangle= b(x) \cdot b(x')$}
    \label{fig:algokernel}
\end{figure}

\subsection{Getting the overlap}
We have a circuit such that 
\begin{equation}
    \mathcal{A}(U)\ket{0}\ket{x} = \ket{b(x)}\ket{x} + \ket{\cdots}\ket{x}^{\perp}
\end{equation}
The goal is to find a circuit such that the expectation of an observable yields $\langle b(x) | b(x') \rangle$. We present such a circuit in \Cref{fig:algokernel}, and prove below that it yields the desired outputs.
The application of $U$ and then of the last Hadamard gate yields the following states:
\begin{align}
    &\frac{1}{\sqrt{2}}(\ket{0}\ket{0}\ket{x} + \ket{1}\ket{0}\ket{x'}) \xrightarrow{I\otimes U} \\
    &\frac{1}{\sqrt{2}}(\ket{0}\ket{b(x)}\ket{0} + \ket{1}\ket{b(x')}\ket{0} + \cdots) \xrightarrow{H\otimes I} \\
    &\frac{1}{2}(\ket{0}(\ket{b(x)}+\ket{b(x')})\ket{0} + \ket{1}(\ket{b(x)}-\ket{b(x')})\ket{0} + \cdots)
\end{align}
Finally the expectation value of $Z\otimes I \otimes \ketbra{0}{0}$ (note it has unit spectral norm) is 
\begin{equation}
   \frac{1}{4}( \lvert \ket{b(x)}+\ket{b(x')} \rvert^2 - \lvert \ket{b(x)}-\ket{b(x')} \rvert^2 ) 
   = \text{Re}(\langle b(x) | b(x') \rangle)\,.
\end{equation}

\subsection{Problem setting: noisy linear model}
In this subsection and the next, we derive bounds on the true risk for a kernel approach to a linear model. 
For simplicity, we will briefly change notation, but $x = b(x)$, $ w= w(\alpha)$ and $y=c_{\alpha}(x)$.
Consider the linear model:
\begin{equation}
    y = x \cdot w
\end{equation}
We are given a $T$-sized noisy Gram matrix $\hat{K} \in \mathbb{R}^{T,T}$ with noisy labels $\hat{Y} \in \mathbb{R}^T$ such that
\begin{equation}
    Y=Xw
\end{equation}
Where $X$ and $Y$ are the underlying exact data and labels, following the linear model with a hidden $w$.
\begin{equation}
    \hat{Y}=Y+E_Y
\end{equation}
Where $E_Y$ is bounded, $\lVert E_Y\rVert_{\infty}\leq \epsilon_y$.
\begin{equation}
    \hat{K}=XX^{\dagger}+E_K
\end{equation}
Where $E_K$ is p.s.d. and each individual entry is bounded by $\epsilon_k$.
We apply kernel ridge regression to the above and aim to bound the true risk.

\subsection{True risk of the kernel ridge regression}
The minimizer of the kernel ridge regression has the following closed-form expression:
\begin{equation}
    a = (\hat{K}+\lambda I)^{-1} \hat{Y}
\end{equation}
Given a new $x'$, we define the true vector of kernel evaluation against training data-points as $F = [k(x_t,x')]_t$. We are only given access to noisy evaluations as $\hat{F} = F+E_F$, where $ \lVert E_F\rVert_{\infty} \leq\epsilon_k$. The goal is to predict the new label $y'$. The squared error in labelling of the new point is:
\begin{equation}
    L =\lvert x'\cdot w - k \cdot a \rvert^2\,.
\end{equation}
The true risk of the kernel ridge regression is the expectation value of this value $L$ over datasets, noise, and new data points. 

\subsection{Upper bound}

We write the squared error as:
\begin{equation}
L = \left|x'^\top w - \hat{F}^\top a\right|^2
\end{equation}

Define $A := (\hat{K} + \lambda I)^{-1}$. Substituting:
\begin{align}
L &= \left|x'^\top w - (Xx' + E_F)^\top A (Xw + E_Y) \right|^2 \\
  &= \left|x'^\top w - x'^\top X^\top A Xw - x'^\top X^\top A E_Y - E_F^\top A Xw - E_F^\top A E_Y \right|^2 \\
  &= \left|x'^\top (I - X^\top A X) w - x'^\top X^\top A E_Y - E_F^\top A Xw - E_F^\top A E_Y \right|^2
\end{align}

We bound each term separately.

\textbf{Bias term:} Let $R := I - X^\top A X$. Assume $\|w\|_2 \leq B$ and $\|x'\|_2 \leq 1$. Then
\begin{equation}
|x'^\top R w| \leq \|R\| \cdot \|x'\| \cdot \|w\| \leq B \cdot \|R\|
\end{equation}
We know from  under $\|E_K\| \leq \epsilon_k$, and that $\lVert k \rVert_{\infty} \leq \kappa$ we may bound:
\begin{equation}
\|R\| \leq \frac{T^2 \kappa}{\lambda^2} \epsilon_k \quad \Rightarrow \quad |x'^\top R w| \leq \frac{ B T^2 \epsilon_k \kappa}{\lambda^2}
\end{equation}

\textbf{Label noise term:}
\begin{equation}
|x'^\top X^\top A E_Y| \leq \|X x'\| \cdot \|A\| \cdot \|E_Y\| \leq \kappa \cdot \frac{1}{\lambda} \cdot \sqrt{T} \epsilon_y = \frac{\kappa \sqrt{T} \epsilon_y}{\lambda}
\end{equation}

\textbf{Kernel noise term:}
\begin{equation}
|E_F^\top A Xw| \leq \|E_F\| \cdot \|A Xw\| \leq \sqrt{T} \epsilon_k \cdot \frac{\|Xw\|}{\lambda} \leq \frac{\sqrt{T} \kappa B \epsilon_k}{\lambda}
\end{equation}

\textbf{Mixed noise term:}
\begin{equation}
|E_F^\top A E_Y| \leq \|E_F\| \cdot \|A E_Y\| \leq \sqrt{T} \epsilon_k \cdot \frac{\sqrt{T} \epsilon_y}{\lambda} = \frac{T \epsilon_k \epsilon_y}{\lambda}
\end{equation}

\textbf{Total Bound:}
\begin{equation}
E \leq \frac{B T^2 \epsilon_k \kappa}{\lambda^2} + \frac{\kappa \sqrt{T} \epsilon_y}{\lambda} + \frac{\sqrt{T} \kappa B \epsilon_k}{\lambda} + \frac{T \epsilon_k \epsilon_y}{\lambda}\,,
\end{equation}
where $E$ is the total error bound.
Then the true risk satisfies:
\begin{equation}
\mathbb{E}[L] \leq E^2
\end{equation}
And $\epsilon$ accuracy of the prediction is guaranteed choosing measurement precision $\epsilon_k$, $\epsilon_y$ and number of training points $T$ such that
\begin{equation}
    \left(\frac{B T^2 \epsilon_k \kappa}{\lambda^2} + \frac{\kappa \sqrt{T} \epsilon_y}{\lambda} + \frac{\sqrt{T} \kappa B \epsilon_k}{\lambda} + \frac{T \epsilon_k \epsilon_y}{\lambda}\right)^2 \leq \epsilon\,.
\end{equation}

\subsection{Conclusion}
The proposed kernel-based learning algorithm is able to PAC learn the concept \Cref{def:parcircCC} under some conditions. 
Because we have $\kappa=1$, when $B$ is at most polynomial, we can find parameters $\epsilon_k, \epsilon_y$ and $T$ that scale polynomially, guaranteeing the efficiency of the learning algorithm.
In particular, if the spectrum is sparse, that there exist a polynomially sized subset of indices $L'=\{l'\}$ such that any $x$ the $b_{l\notin L'}(x)=0$, then $B \in \poly(n)$.

\section{Flipped concept and connection to RFF}
The flipped concept, where the input and the index of the concept class are inverted as, is in contrast to the concept class $\mathcal{C}^U$ we study easy to learn.
We define it as follows,
$$ \bar{\mathcal{C}} \coloneqq \{c_x : \alpha\in\mathbb{R}^d \rightarrow \sum_l b_l(x)e^{il\cdot \alpha}\}_{x\in\{0,1\}^*}\,.$$
This corresponds to a scenario where a quantum circuit has some fixed, potentially unknown gates and some parametrized gates. Although this concept class looks similar to the one we studied earlier, for this one it is easy to see that unlike $\mathcal{C}^U_{n,\log{n}}$, the concepts here are in $\mathsf{P/poly}$, where the advice is the polynomially sized list of $b_l(x)$. 
Therefore, the arguments we make about the hardness of $\mathcal{C}_{n,\log{n}}$ do not apply. 

In fact, this concept is classically efficiently learnable by a simple Fourier analysis of the data.
This is typically the scenario of quantum neural networks and quantum kernels, which have been dequantized by techniques like Random Fourier Features \cite{sweke_potential_2023, thabet_when_2024,sahebi_dequantization_2025}. 
In fact, using the circuit proposed and sampling from the frequency register after post-selection on the circuit register, yields the optimal distribution to approximate the circuit using RFF, as it satisfies perfectly the alignment condition.

\end{document}